\documentclass[12pt]{article}
\newcommand{\version}{\today}

\usepackage{amsthm,amsfonts,amsmath, amscd}
\usepackage{color}

\swapnumbers
                              
\theoremstyle{plain}
\newtheorem{thm}{THEOREM}[section]
\newtheorem{lm}[thm]{LEMMA}

\newtheorem{prop}[thm]{PROPOSITION}

\theoremstyle{definition}
\newtheorem{defi}[thm]{DEFINITION}
\newtheorem{remark}[thm]{Remark}
\theoremstyle{remark}
\newcommand{\upchi}{\raise1pt\hbox{$\chi$}}
\newcommand{\R}{{\mathord{\mathbb R}}}
\newcommand{\C}{{\mathord{\mathbb C}}}

\newcommand{\N}{{\mathord{\mathbb N}}}

\newcommand{\hn}{{\mathord{\widehat{n}}}}

\newcommand{\tr}{{\rm Tr}}
\renewcommand{\|}{{\Vert}}

\numberwithin{equation}{section}
\newcommand{\pn}{{\cal P}}

\newcommand{\un}{{\rm 1\kern -2.5pt l}}

    \newcommand{\one}{{\bf 1}}

\begin{document}
\markboth{\scriptsize{CL \version}}{\scriptsize{CL September 15,  2017}}
\def\mn{{\bf M}_n}
\def\hn{{\bf H}_n}
\def\pn{{\bf P}_n}
\def\hmnp{{\bf H}_{mn}^+}
\def\H{{\mathcal H}}
\def\EA{\mathbb{E}_{\mathcal{A}}}

\title{{\sc  Some trace inequalities for  exponential and logarithmic 
functions}}
\author{
\vspace{5pt}  Eric A. Carlen$^{1}$  and Elliott H. Lieb$^{2}$ \\
\vspace{5pt}\small{$^{1}$ Department of Mathematics, Hill Center,}\\[-6pt]
\small{Rutgers University,
110 Frelinghuysen Road
Piscataway NJ 08854-8019 USA}\\
\vspace{5pt}\small{$^{2}$ Departments of Mathematics and Physics, Jadwin
Hall, Princeton University}\\[-6pt]
\small{ Washington Road, Princeton, NJ  08544.}\\[-6pt]
 }

\date{April 15, 2018}
\maketitle

\footnotetext [1]{Work partially supported by U.S.
National Science Foundation grant  DMS 1501007.}

\footnotetext [2]{Work partially supported by U.S. National
Science Foundation
grant PHY 1265118 \hfill\break
\copyright\, 2017 by the authors. This paper may be reproduced, in its
entirety, for non-commercial purposes.
}

\begin{abstract}
\footnotesize{
Consider a function $F(X,Y)$ of pairs of positive matrices with values in 
the positive matrices such that whenever $X$ and $Y$ commute 
$F(X,Y)= X^pY^q.$ Our first main result gives  conditions on $F$ such 
that $\tr[ X \log (F(Z,Y))] \leq \tr[X(p\log X + q \log Y)]$ for all 
$X,Y,Z$ such that   $\tr Z =\tr X$. (Note that $Z$ is absent from the right 
side of the inequality.) We  give several examples of functions $F$ to 
which the theorem applies.  

Our theorem allows us to give simple proofs of the well known logarithmic 
inequalities of Hiai and Petz and several new generalizations of them which 
involve three variables $X,Y,Z$ instead of just $X,Y$ alone. The investigation of these logarithmic inequalities is closely connected with three quantum relative entropy functionals: The standard Umegaki quantum relative entropy $D(X||Y) = \tr [X(\log X-\log Y])$, and two others, the Donald relative entropy
$D_D(X||Y)$, and the Belavkin-Stasewski relative entropy $D_{BS}(X||Y)$. They are known to satisfy
$D_D(X||Y) \leq D(X||Y)\leq D_{BS}(X||Y)$. We prove that the Donald relative entropy 
provides the sharp upper bound, independent of $Z$ on  
$\tr[ X \log (F(Z,Y))]$  in a number of cases in which $F(Z,Y)$
is homogeneous of degree $1$ in $Z$ and $-1$ in $Y$. 
We also investigate the Legendre transforms in $X$
of $D_D(X||Y)$ and $D_{BS}(X||Y)$, and show how our results for these lead to new refinements of the
Golden-Thompson inequality. 
}

\end{abstract}

\medskip
\centerline{Key Words: trace inequalities, quantum relative entropy, convexity}

\section{Introduction}  

Let $\mn$ denote the set of complex $n\times n$ matrices. Let $\pn$ and 
$\hn$ denote the subsets of $\mn$ consisting of strictly positive 
and self-adjoint matrices respectively.  For $X,Y\in \hn$,   $X\geq Y$ to indicate that 
$X-Y$ is positive semi-definite; i.e., in  the closure of
$\pn$, and  $X>Y$ 
indicates that $X\in \pn$.

Let $p$ and $q$ be non-zero real numbers. There are many functions $F:\pn \times \pn \to \pn$
such that $F(X,Y) = X^pY^q$ whenever $X$ and $Y$ compute. For example, 
\begin{equation}\label{firstalt}
F(X,Y) = X^{p/2}Y^qX^{p/2} \quad{\rm or}\qquad F(X,Y) = Y^{q/2}X^pY^{q/2}\ .
\end{equation}
Further examples can be constructed using  geometric means:
For positive $n\times n$ matrices $X$ and $Y$, 
and $t\in [0,1]$, the {\em $t$-geometric mean of $X$ and $Y$,} denoted by
$X\#_t Y$, is defined by Kubo and Ando  \cite{KA80}  to be
\begin{equation}\label{geomean}
X\#_t Y := X^{1/2}(X^{-1/2}YX^{-1/2})^t 
X^{1/2}\ .
\end{equation}
The geometric mean for $t= 1/2$ was initially defined  and studied by Pusz 
and 
Woronowicz \cite{PW75}. The formula \eqref{geomean} makes sense
for all $t\in \R$ and it  has a natural geometric meaning \cite{S84}; see 
the discussion around Definition \ref{gmdef}   and in Appendix C. 
Then for all $r>0$ and all $t\in(0,1)$,
\begin{equation}\label{thirdalt}
F(X,Y) = X^r\#_t Y^r
\end{equation} is such a function with $p = r(1-t)$ and $q = rt$. Other examples will be considered below. 

If $F$ is such a function,  then
$\tr[X \log F(X,Y)] = \tr[X(p\log X + q\log Y)]$ whenever $X$ and $Y$ commute.
We are interested in conditions on $F$ that guarantee either
\begin{equation}\label{logone}
\tr[X \log F(X,Y)]  \geq  \tr[X(p\log X + q\log Y)]
\end{equation}
or 
\begin{equation}\label{logtwo}
\tr[X \log F(X,Y)] \leq \tr[X(p\log X + q\log Y)]
\end{equation}
for all $X,Y\in \pn$.   Some examples of such inequalities are known: 
Hiai and Petz \cite{HP93} proved that
\begin{equation} \label{HPpair} 
\frac{1}{p}\tr [X \log (Y^{p/2} X^p Y^{p/2} ) ]\  \leq \ \tr [ X (\log X + \log Y )] \
\leq \
 \frac{1}{p}\tr 
 [X \log (X^{p/2} Y^{p} X^{p/2} ) ] \ ,
\end{equation}
for all $X,Y > 0$ and all $p>0$.  Replacing $Y$ by $Y^{q/p}$ 
shows that for $F(X,Y) = X^{p/2}Y^qX^{p/2}$, \eqref{logone} is valid, while for 
$F(X,Y) =  Y^{q/2}X^pY^{q/2}$, \eqref{logtwo} is valid: Remarkably, the effects of 
non-commutativity go in different directions in these two examples. Other examples involving functions $F$ of the form
\eqref{thirdalt} have been proved by Ando and Hiai \cite{AH94}. 

Here we prove several new inequalities of this type, and we also strengthen the results cited above by bringing in a {\em third} operator $Z$:
For example, Theorem~\ref{clA}  says that for all postive $X$, $Y$ and $Z$ such that $\tr[Z] = \tr[X]$, 
\begin{equation}\label{logtwoB}
\frac{1}{p}\tr [X \log (Y^{p/2} Z^p Y^{p/2} ) ]\  \leq \ \tr [ X (\log X + \log Y )]
\end{equation}
with strict inequlaity if $Y$ and $Z$ do not commute.
If $Y$ and $Z$ do commute, the left  side of \eqref{logtwoB} is simply 
$\tr [ X (\log Z + \log Y )]$, and  the inequality \eqref{logtwoB} would then  follow from the inequality $\tr[X\log Z] \leq \tr[X\log X]$ for all positive 
$X$ and $Z$ with $\tr[Z] = \tr[X]$.  Our result shows that this persists in the non-commutative case, and we obtain similar results for
other choices of $F$, in particular for those defined in terms of gemetric means. 

One of the reasons that inequalities of this sort are of interest is their connection 
with quantum relative entropy. 
By taking $Y =W^{-1}$, with $X$ and $W$ both having unit trace, so that 
both $X$ and $W$ are density matrices,
the middle quantity in (\ref{HPpair}), $\tr [ X (\log X - \log W )]$, is 
the Umegaki relative entropy of $X$ with respect to $W$. 
 Thus 
(\ref{HPpair}) provides upper and lower bounds on the  
relative entropy.

There is another source of interest in the inequalities (\ref{HPpair}), which Hiai and Petz refer to as 
{\em logarithmic inequalities}. As they point out, logarithmic inequalities are dual, via the Legendre transform, to certain {\em exponential 
inequalities} 
related to the Golden-Thompson inequality. Indeed, {\em the quantum Gibbs 
variational principle} states that
\begin{equation}\label{gvp}
\sup\{ \tr[XH] - \tr[X(\log X - \log W)]\ :\ 
X\geq 0\ ,\ \tr[X] = 1\} = \log(\tr[e^{H + \log W}])   \ ,
\end{equation}
for all self-adjoint $H$ and all non-negative $W$.
(The quantum Gibbs variational principle is a 
direct 
consequence of the Peierls-Bogoliubov inequality, see Appendix~\ref{PBGV}.)  

It follows immediately from (\ref{HPpair}) and (\ref{gvp}) that 
\begin{equation}\label{gvp2}
   \sup\{ \tr[XH] - \tr[X(\log(X^{1/2}W^{-1}X^{1/2})]\ :\  X\geq 0\ ,\ 
\tr[X] = 1\}  \leq \log(\tr[e^{H + \log W}]) \ .
\end{equation}
The left side of (\ref{gvp2}) provides a lower bound for 
$\log(\tr[e^{H + \log W}])$ in terms of a Legendre transform, which, 
unfortunately, 
cannot be  
evaluated explicitly.  

An  alternate use of the inequality on the right in \eqref{HPpair}
does yield an explicit lower bound on $\log(\tr[e^{H + \log W}]) $ in 
terms of a geometric mean of $e^H$ and $W$.  This was done in 
\cite{HP93}; the bound is 

\begin{equation}\label{HPco}
\tr[(e^{rH} \#_{t}e^{rK})^{1/r}]  \leq \tr[  e^{(1-t)H + t  K}],
\end{equation}
which is valid for all self adjoint $H,K$, and all $r>0$ and $t\in [0,1]$.
Since the Golden-Thompson inequality is $\tr[  e^{(1-t)H + 
t   K}] 
\leq \tr[  e^{(1-t)H}e^{t K}]$, \eqref{HPco} is viewed in 
\cite{HP93} as a {\em complement to the Golden-Thompson inequality.}

Hiai and Petz show \cite[Theorem 2.1]{HP93} that the inequality 
\eqref{HPco}  is equivalent to  the inequality on the right in 
\eqref{HPpair}.
One direction in proving the equivalence, starting from (\ref{HPco}), is a 
simple differentiation 
argument; differentiating (\ref{HPco}) at $t=0$ yields the result. 
While the inequality on the left in 
(\ref{HPpair}) is relatively simple to prove, the one on the right appears 
to be deeper and more difficult to prove, from the perspective of 
\cite{HP93}. 

In our paper we prove a number of new inequalities,  some of which 
strengthen and extend (\ref{HPpair}) and 
(\ref{HPco}).  Our results show,  in particular, that the 
geometric mean provides a natural bridge between the pair of  
inequalities (\ref{HPpair}).  This perspective  yields a fairly simple 
proof of the deeper inequality on the right of \eqref{HPpair}, and 
thereby places the appearance of the geometric mean in (\ref{HPco}) in a 
natural context.

Before stating our results precisely, we  recall the notions of { operator concavity} and { operator convexity}. 
A function $F: \pn  \to \hn$ is {\em concave} in case for all $X,Y\in \pn$ and all $t \in [0,1]$,
$$F((1-t)X + t Y) - (1-t)F(X) - t F(Y) \in \pn\ ,$$
and $F$ is {\em convex} in case $-F$ is concave.  For example, $F(X) := X^p$ is concave for $p\in [0,1]$
as is $F(x) := \log X$.  

A function $F: \pn\times \pn \to \hn$ is {\em jointly concave} in case for 
all $X,Y,W,Z\in \pn$ and all $t \in [0,1]$
$$F((1-t)X + t Y, (1-t)Z + t W) - (1-t)F(X,Z) - t F(Y,W) \in \pn\ ,$$
and $F$ is {\em jointly convex} in case $-F$ is {\em jointly concave}. 
{\em Strict} concavity or convexity means that the left side is never zero 
for any  $t\in(0,1)$ unless $X=Y$ and $Z=W$. A
particularly well-known and important example is provided by the 
generalized geometric means.
By a theorem of Kubo and Ando \cite{KA80}, for each $t\in [0,1]$, $F(X,Y) 
:= 
X\#_t Y$ is jointly concave in $X$ and $Y$. 
Other examples of jointly concave functions  are discussed below.

Our first main result is the following:

\begin{thm}\label{main1}  Let $F:\pn\times \pn\to \pn$ be such that:

\smallskip
\noindent{\it (1)} For each fixed $Y\in \pn$, $X \mapsto F(X,Y)$ is concave, and for all $\lambda>0$,
$F(\lambda X,Y) = \lambda F(X,Y)$.

\smallskip
\noindent{\it (2)} For each $n\times n$ unitary matrix $U$, and each $X,Y\in \pn$, 
\begin{equation}\label{cond1}
F(UXU^*,UYU^*) = UF(X,Y)U^*\ .
\end{equation}

\smallskip
\noindent{\it (3)} For some  $q\in \R$,  if $X$ and $Y$ commute 
then $F(X,Y) = XY^q$.

\smallskip

Then,  for all $X,Y,Z\in \pn$ such that $\tr[Z] = \tr[X]$,
\begin{equation}\label{mainineq}
\tr[ X \log (F(Z,Y))] \leq \tr[X(\log X + q \log Y)]\ .
\end{equation}
 If, moreover, $X \mapsto F(X,Y)$ is strictly concave, then the
inequality in (\ref{mainineq}) is strict when $Z$ and $Y$ do not 
commute.
\end{thm}

\begin{remark} Notice that \eqref{mainineq} has three variables on the 
left, but only two on the right. The third variable $Z$ is related to $X$ 
and $Y$ {\em only} through the constraint  $\tr[ Z]= \tr [X]$. 
\end{remark}

Different choices for the function $F(X,Y)$ yield different corollaries.   
For our first corollary, we take 
the  function
$\displaystyle{ F(X,Y) = \int_0^\infty \frac{1}{\lambda 
+Y}X\frac{1}{\lambda +Y}{\rm d}\lambda}$,  which evidently satisfies the conditions of 
Theorem~\ref{main1} with $q = -1$.
 We obtain, thereby, the following inequality:

 \begin{thm}\label{clC} Let  $X,Y,Z\in \pn$ be such that $\tr[Z] = \tr[X]$,  Then 
\begin{equation}\label{mainC}
\tr\left[ 
X \log \left( \int_0^\infty \frac{1}{\lambda +Y}Z\frac{1}{\lambda +Y}{\rm d}\lambda
  \right)
 \right] 
  \leq \tr[X( \log X -  \log Y)]\ .
\end{equation}
\end{thm}

Another simple application can be made to the function $F(X,Y) = Y^{1/2}XY^{1/2}$, however in this case, an adaptation of method of  proof 
of Theorem~\ref{main1} yields a more general result for the two-parameter family of functions  $F(X,Y) = Y^{p/2}X^pY^{p/2}$ for all $p>0$

%

\begin{thm}\label{clA} For all $X,Y,Z\in \pn$ such that $\tr[Z] = \tr[X]$, and all $p > 0$,
\begin{equation}\label{mainA}
\tr[ X \log (Y^{p/2}Z^pY^{p/2}))] \leq \tr[X(\log X^p +  \log Y^p)]\ .
\end{equation}
The inequality in 
(\ref{mainA}) is strict unless $Z$ and $Y$ commute,
\end{thm}

Specializing to the case $Z = X$, (\ref{mainA}) reduces to the inequality 
on the left in (\ref{HPpair}).  Theorem~\ref{clA} thus extends the
inequality of \cite{HP93} by inclusion of the third variable $Z$, and specifies the cases of equality there.

\begin{remark}\label{trcase} If $Z$ does commute with $Y$, (\ref{mainA}) reduces to $\tr[X\log Z] \leq \tr[X\log X]$ 
which is well-known to be true under the condition $\tr[Z] = \tr[X]$, with 
equality if and only if $Z= X$ .
\end{remark}

We also obtain results for 
 the {\em two parameter} family of functions
$$F(X,Y) = Y^r\#_s X^r $$ with $s\in [0,1].$ and $r>0$.    In this case, when $X$ and $Y$ commute, $F(X,Y) = X^pY^q$ with 
\begin{equation}\label{homd}
p= rs \quad{\rm and} \quad q = r(1-s)\ .
\end{equation}
It would be possible to deduce at least some of these results directly from Theorem~\ref{main1} is we knew that, for example, $X\mapsto
Y^2\#_{1/2}X^2 = Y(Y^{-1}X^2Y^{-1})^{1/2}Y$ is concave in $X$. While we have no such result, it turns out that we can use Theorem~\ref{clA}
to obtain the following:


\begin{thm}\label{clB} Let $X,Y,Z\in \pn$ be such that $\tr[Z] = \tr[X]$. 
Then for all $s\in [0,1]$ and all $r>0$,
\begin{equation}\label{mainB}
\tr[ X \log (Y^r\#_s Z^r))] \leq \tr[X(s \log X^r +  (1-s)\log Y^r)]\ .
\end{equation}
For $s\in (0,1)$, when $Z$ does not commute with $Y$, the inequality is strict. 
\end{thm}

The case in which $Z=X$ is proved in \cite{AH94} using log-majorization methods. The inequality \eqref{mainB} is an identity at $s=1$. As we shall show, differentiating it 
at $s=1$ in the case $Z =X$ yields the inequality on the right in \eqref{HPpair}.  Since the 
geometric mean inequality \eqref{mainB} is a consequence of our 
generalization of the inequality on the left in \eqref{HPpair}, 
 this derivation shows how the 
geometric means construction `bridges'  the pair of inequalities 
(\ref{HPpair}).

Theorems~\ref{clC}, \ref{clA} and \ref{clB} provide infinitely many new lower bounds on the Umegaki relative entropy. -- one for each choice of $Z$.
The trace functional on the right side of (\ref{HPpair}) bounds the Umegaki relative entropy from above, and  in many ways 
better-behaved than the trace functional on the left, or any of the individual new lower bounds.
By a theorem 
of Fujii and Kamei \cite{FK88}
$$X,W \mapsto X^{1/2}\log(X^{1/2} W^{-1}X^{1/2})X^{1/2}$$ 
is jointly convex as a function from $\pn\times \pn$ to $\pn$, and then as a trivial consequence,
$$X,W \mapsto \tr[X\log(X^{1/2} W^{-1}X^{1/2})]$$ 
is jointly convex. When $X$ and $W$ are density matrices, $\tr[X\log(X^{1/2} W^{-1}X^{1/2})] =: D_{BS}(X||W)$
is the {\em Belavkin-Stasewski relative entropy} \cite{BS82}.   The joint 
convexity of the Umegaki relative entropy is a Theorem of Lindblad 
\cite{Lind74}, who deduced it as a direct consequence of  the main  
concavity theorem in \cite{L73}.

A seemingly small change in the arrangement of the operators -- $X^{1/2}W^{-1}X^{1/2}$ replaced with 
$W^{-1/2}XW^{-1/2}$ -- obliterates convexity; 
\begin{equation}\label{nonco}
X,W \mapsto \tr[X\log(W^{-1/2}XW^{-1/2})]
\end{equation}
is not jointly convex, and even worse,
 the function
$W \mapsto \tr[X\log(W^{-1/2}XW^{-1/2})]$
is {\em not convex} for all fixed $X\in \pn$. 
Therefore, although the function in (\ref{nonco}) agrees with the Umegaki relative entropy when 
$X$ and $W$ commute, its lack of convexity makes it unsuitable for 
consideration as a relative entropy functional. We discuss the failure of 
convexity at the end of Section 3.

However, Theorem~\ref{clA} provides a remedy by introducing a third 
variable $Z$ with respect to which we can maximize. The resulting 
functional is still bounded above by the Umegaki relative entropy: that is,
for all density matrices $X$ and $W$, 
\begin{equation}\label{grelU}
\sup\{ \tr[X\log(W^{-1/2}ZW^{-1/2})]\ :\ Z\geq 0\ , \tr[Z] \leq1\ \}  \leq 
D(X||W)\ .
\end{equation}
One might hope that the left side is a jointly convex function of 
$X$ and $W$, {\em which does turn out to be the case.}  In fact, the left 
hand side is a quantum relative entropy originally introduced by Donald 
\cite{Do86}, through a quite different formula.  Given any orthonormal basis 
$\{u_1,\dots,u_n\}$ of $\C^n$,
define a ``pinching'' map $\Phi:\mn\to \mn$ by defining $\Phi(X)$ to be the diagonal matrix whose $j$th diagonal entry is
$\langle u_j,Xu_j\rangle$. Let $\mathcal{P}$ denote the sets of all such 
pinching operations. For density matrices $X$ and $Y$,
the {\em Donald relative entropy}, $D_D(X||Y)$ is defined by
\begin{equation}\label{dondef}
D_D(X||Y) = \sup\{  D(\Phi(X)||\Phi(Y))\ :\ \Phi\in \mathcal{P}\}\ . 
\end{equation}
Hiai and Petz \cite{HP93} showed that for all density matrices $X$ and all $Y\in \pn$,
\begin{equation}\label{donfor}
D_D(X||Y) = \sup\{  \tr[XH] - \log\left(\tr[e^HY]\right) \ :\ H\in \hn\}\ ,
\end{equation}
arguing as follows. 
Fix any orthonormal basis $\{u_1,\dots,u_n\}$ of $\C^n$.  Let $X$ be any density matrix and let $Y$ be any positive matrix.
Define $x_j = \langle u_j,X u_j\rangle$ and $y_j = \langle u_j,Y u_j\rangle$ for $j=1,\dots,n$.   For 
$(h_1,\dots,h_n)\in \R^n$, 
define $H$ to be the self-adjoint operator given by $Hu_j = h_ju_j$, $j=1,\dots,n$. 

Then by the classical Gibb's variational principle.
\begin{eqnarray*}\sum_{j=1}^n x_j (\log x_j  - \log y_j ) &=& \sup\left\{ \sum_{j=1}^n x_j h_j  - \log\left(\sum_{j=1}^n e^{h_j} y_j\right)\ :\ (h_1,\dots,h_n)\in \R^n\right\}\\
&=& 
\sup\left\{ \tr[XH]  - \log\left(\tr[e^HY]\right)\ :\ (h_1,\dots,h_n)\in \R^n\right\}
\end{eqnarray*}
Taking the supremum over all choices of the orthonormal basis yields \eqref{donfor}. For our purposes, a variant of 
\eqref{donfor} is useful:
\begin{lm} For all density matrices $X$, and all  $Y\in \pn$, 
\begin{equation}\label{donfor2}
D_D(X||Y) = \sup\{  \tr[XH]  \ :\ H\in \hn\ ,  \tr[e^HY] \leq 1 \}\ .
\end{equation}
\end{lm}

\begin{proof} Observe that we may add a constant to $H$ without changing $\tr[XH] - \log\left(\tr[e^HY]\right)$, and thus in taking  the supremum in \eqref{donfor} we may restrict our attention to $H\in \hn$ such that $\tr[e^HY] =1$. Then
$\tr[XH] - \log\left(\tr[e^HY]\right)  = \tr[XH]$ and the constraint in \eqref{donfor2} is satisfied. Hence the supremum in \eqref{donfor} is no larger than the supremum in \eqref{donfor2}. Conversely, if $ \tr[e^HY] \leq 1$, then 
$$\tr[XH] \leq \tr[XH]  - \log\left(\tr[e^HY]\right)\ ,$$
and thus the supremum in \eqref{donfor2} is no larger than the supremum in \eqref{donfor}. 
\end{proof}

By the joint convexity of the Umegaki relative entropy, for each $\Phi\in 
{\mathcal P}$, $D(\Phi(X)||\Phi(Y))$ is jointly convex in $X$ and $Y$, and 
then since the supremum of a family of convex functions is convex, the 
Donald relative entropy $D_D(X||Y)$ is jointly convex. Making the change of 
variables $Z = W^{1/2}e^HW^{1/2}$ in \eqref{grelU}, one sees that the 
supremum in 
\eqref{donfor} is exactly the same as the supremum in \eqref{donfor2}, and thus for all density matrices $X$ and $W$, 
$D_D(X||W) \leq D(X||W)$ which can also be seen as a consequence of the joint convexity of the Umegaki relative entropy. 

Theorem~\ref{clC} and Theorem~\ref{clB} give two more lower bounds to the Umegaki relative entropy 
 for density matrices $X$ and $Y$, namely
\begin{equation}
\sup_{Z\in \pn, \tr[Z]= \tr[X]}\left\{\   \tr\left[ 
X \log \left( \int_0^\infty \frac{1}{\lambda +Y}Z\frac{1}{\lambda +Y}{\rm 
d}\lambda
  \right)
 \right]
\right\}
\end{equation}
and 
\begin{equation}
\sup_{Z\in \pn, \tr[Z]= \tr[X]}\left\{\  \tr[ X \log (Y^{-1}\#_{1/2}Z)^2]  
 \right\} 
\end{equation}
Proposition~\ref{allsame} shows that both of the supremums are equal to 
$D_D(X||Y)$.

Our next results concern the partial Legendre transforms of the three 
relative entropies $D_D(X||Y)$, $D(X||Y)$ and 
$D_{BS}(X||Y)$. For this,  it is natural to consider them as functions  on $\pn\times\pn$, and not only on density matrices. 
The natural extension of the Umegaki relative 
entropy functional to
$\pn\times \pn$ is
\begin{equation}\label{vnent2}
D(X||W) :=  
\tr[X(\log X - \log W)] + \tr[W] - \tr[X]\ 
\end{equation}
It is homogeneous of degree one in $X$ and $W$ and,  with this definition, 
$D(X||Y) \geq 0$ with equality only in case $X= W$, which is a consequence 
of Klein's inequality, as discussed in Appendix A. 

The natural extension of the Belavkin-Stasewski relative entropy functional to
$\pn\times \pn$ is
\begin{equation}\label{bsre}
D_{BS}(X||W) = \tr[X\log(X^{1/2}W^{-1}X^{1/2})] + \tr[W] - \tr[X]\ .
\end{equation}
Introducing $Q := e^H$,  the supremum in (\ref{donfor2}) is
$$\sup\{ \tr[X\log Q]\ :\ Q\geq 0\ , \tr[WQ] \leq 1\ \} \ ,$$
and the extension of the Donald relative entropy to $\pn\times \pn$ is  
\begin{equation}\label{geom31}
D_{D}(X||W) = \sup_{Q > 0}\{ \tr[X\log Q] \ :\ \tr[ WQ]  \leq\tr[X]\ \} 
+\tr[W] - \tr[X] \
\end{equation}

To avoid repetition, it is useful to note that all three of these  
functionals are examples of quantum relative entropy  functionals in the 
sense of satisfying the following axioms. This axiomatization differs 
from many others, such as the ones in \cite{Do86} and \cite{Ha16}, which 
are designed to single out the Umegaki relative entropy. 

\begin{defi}\label{qred}   A {\em quantum relative entropy}  is a function 
$R(X||W)$ on $\pn\times \pn$ with values in $[0,\infty]$ such that 
\smallskip 

\smallskip
{\it (1)} $X,Y \mapsto R(X||W)$ is jointly convex.

\smallskip
{\it (2)} For all $X,W\in \pn$ and all $\lambda>0$,  
$R(\lambda X,\lambda W) = \lambda R(X,W)$
and
\begin{equation}\label{scaling2A}
R(\lambda X, W) = \lambda R(X,W) + \lambda\log \lambda\tr[X] + 
(1-\lambda)\tr[W]\ .
\end{equation}

\smallskip
{\it (3)} If $X$ and $W$ commute, $R(X||W) = D(X||W)$. 
\end{defi}

The definition does not include the requirement that $R(X||W) \geq 0$ with equality if and only if $X = W$ because this follows 
directly from {\it (1)}, {\it (2)}  and {\it (3)}:

\begin{prop}\label{pinsker}  Let $R(X||W)$ be any quantum relative entropy. Then
\begin{equation}\label{pinsk}
R(X||W) \geq \tfrac12 \tr[X]  \left\Vert \frac{X}{\tr[X]} 
-\frac{W}{\tr[W]}\right\Vert_1^2
\end{equation}
where $\|\cdot \|_1$ denotes the trace norm.
\end{prop}

The proof is given towards the end of Section 3. It is known for the Umegaki relative entropy \cite{HOT81}, but the proof uses only the properties {\it (1)}, {\it (2)}  and {\it (3)}.

The following pair of inequalities summarizes the relation among the 
three relative entropies. For  all  $X,W\in \pn$,
\begin{equation}\label{triple}
D_D(X||W) \leq D(X||W)  \leq   D_{BS}(X||W)\ .
\end{equation}
These inequalities will imply  a corresponding pair of inequalities for 
the partial Legendre transforms in $X$. 

\begin{remark}The partial Legendre transform of the relative entropy, which 
figures in the Gibbs variational principle, is in many ways better behaved 
than the full Legendre transform. Indeed the Legendre   transform $F^*$ of 
a function $F$ on $\R^n$ that is convex and homogenous of degree one always 
has the form 
$$F^*(y) = \begin{cases} 0 & y\in C\\ \infty & y \notin C\end{cases}
$$
for some convex set $C$ \cite{R70}.   The set $C$ figuring in the full Legendre 
transform of the Umegaki relative entropy  was 
first computed by Pusz and Woronowicz \cite{PW78}, and somewhat more 
explicitly by Donald  in \cite{Do86}. 
\end{remark}

Consider any function $R(X||Y) $ on $\pn \times \pn $ that is convex and 
lower semicontinuous in $X$.  There are two natural partial Legendre transforms 
that are related to each other, namely  $\Phi_R(H,Y)$ and 
$\Psi_R(H,Y)$ defined by
\begin{equation}\label{PhiRd}
\Phi_R(H,Y)  = \sup_{X\in \pn}\{ \tr[XH] - R(X||Y)\ :\  \tr[X] =1 \}
\end{equation}
and 
\begin{equation}\label{PsiRd}
\Psi_R(H,Y)  = \sup_{X\in \pn}\{ \tr[XH] - R(X||Y) \}
\end{equation}where 
$H\in \hn$ is the conjugate variable to $X$.

For example, let $R(X||Y) = D(X||Y) $ , the Umegaki relative entropy.  
Then, by the Gibbs variational principle, 
\begin{equation} \label{phiu}
 \Phi(H,Y) =  1 - \tr[Y] + \log( \tr e^{H+ \log Y} ) 
\end{equation}
and 
\begin{equation} \label{psiu}
 \Psi(H,Y) =  \tr e^{H+ \log Y}   - \tr Y\ .
\end{equation}

\begin{lm}\label{logconc}  Let $R(X||Y)$ be any function on $\pn 
\times \pn $ that is convex and 
lower semicontinuous in $X$, and which satisfies the scaling relation 
\eqref{scaling2A}.
Then for all $H\in \hn$ and all $Y\in \pn$. 
\begin{equation}\label{dualrel}
\Psi_R(H,Y) = e^{ \Phi_R(X,Y) + \tr[Y] -1}  - \tr[Y]\ .
\end{equation}
\end{lm}

This simple relation between the two Legendre transforms is a consequence of scaling, and hence the corresponding relation holds for {\em any} quantum relative entropy.

Consider the Donald relative entropy and define
\begin{equation}\label{geom33RA}
\Psi_D(H,Y) :=  \sup_{X>0}  \{  \tr[XH] - D_D(X||Y) \ \}\ ,
\end{equation}
and
\begin{equation}\label{geom33}
\Phi_D(H,Y) :=  \sup_{X>0, \tr[X] =1}  \{  \tr[XH] - D_D(X||Y) \ \}\ ,
\end{equation}

In Lemma~\ref{mmlem}, we prove the following analog of \eqref{phiu}: For $H\in \hn$ and $Y\in \pn$,
\begin{equation}\label{MM55}
\Phi_D(H,Y) = 1 - \tr[Y]  + \inf\big\{  \lambda_{\max}\left( H - \log Q\right)  
\ :\ Q\in \pn\ , \tr[QY] \leq 1\ \big\}\ 
\end{equation}
where for any self-adjoint operator $K$, $\lambda_{\max}(K)$ is the largest 
eigenvalue of $K$, and we prove that $\Phi_D(H,Y)$  is concave in $Y$. As a consequence of this we prove in Theorem~\ref{expconcavbe} that
for all $H\in \hn$, the function
\begin{equation}\label{legscale2}
Y \mapsto \exp\left( \inf_{Q>0, \tr[QY] \leq 1} \lambda_{\max}\left( H - \log Q\right)\right)  
\end{equation}
is concave on $\pn$.  Moreover,  for all  $H,K\in \hn$, 
\begin{equation}\label{geom52}
\log(\tr[e^{H+K}] )\leq \  \inf_{Q>0, \tr[Qe^K] \leq 1} 
\lambda_{\max}\left( H - \log Q\right)  \  \leq  \log( \tr[e^{H} e^{K}]).
\end{equation}
These inequalities  improve upon 
the Golden-Thompson inequality.  Note that by Lemma~\ref{logconc}, \eqref{psiu} and 
\eqref{MM55}, the inequality on the left in \eqref{geom52} is equivalent to 
$\Psi(H,Y) \leq \Psi_D(H,Y)$, which in turn is equivalent under the Legendre transform to 
$D_D(X||Y) \leq D(X||Y)$.

The inequality on the right in \eqref{geom52} arises through the simple of choice  
$Q = e^H/\tr[Y e^H]$ in the variational formula for $\Psi_D(H,Y)$.  
The $Q$ chosen here is optimal {\it only} when $H$ and $Y$ commute. 
Otherwise,
there is a better choice for $Q$, which we shall identify in section  4, 
and which will lead to a tighter upper bound.  
In section 4 we shall also discuss the Legendre transform of the Belavkin-Staszewski 
relative entropy and form this we derive further refinements of the Golden Thompson 
inequality. Finally, in Theorem~\ref{HPC2} we prove a sharpened form of \eqref{HPco}, the complementary  Golden-Thompsen inequality of Hiai and Petz, incorporating a relative entropy remainder term. Three appendices collect background material for the convenience of the reader.

\section{Proof of Theorem~\ref{main1} and Related Inequalities }

\begin{proof}[Proof of Theorem~\ref{main1}]  Our goal is to prove that
for all $X,Y,Z\in \pn$ such that $\tr[Z] = \tr[X]$.
\begin{equation}\label{mainineqA}
\tr[ X \log (F(Z,Y))] \leq \tr[X(\log X + q \log Y)]\ ,
\end{equation}
whenever $F$ has the properties {\it (1)}, {\it (2)} and {\it (3)} listed in 
the statement of Theorem~\ref{main1}. 
By the homogeneity specified in {\it (3)}, we may assume without loss of 
generality that $\tr[X] = \tr[Z] = 1$. Note that
(\ref{mainineqA}) is equivalent to
\begin{equation}\label{mainineqB}
\tr\left[ X\left( \log (F(Z,Y)) - \log X -q \log 
Y)\right)\right ] \leq 0\ ,
\end{equation}
By  the Peierls-Bogoliubov inequality (\ref{PB}),  it suffices to prove that
\begin{equation}\label{mainineqC}
\tr\left[ \exp\left(  \log (F(Z,Y))  -q \log 
Y)\right)\right ] \leq 1\ .
\end{equation}
Let $\mathcal{J}$ denote an arbitrary finite index set with cardinality 
$|\mathcal{J}|$.  Let $\mathcal{U} = \{ U_1, \dots, U_{|\mathcal{J}|}\}$ be 
any set of unitary matrices each of which commutes with $Y$. Then for each 
$j\in \mathcal{J}$, by {\it (2)}
\begin{eqnarray}\label{mainY8}
\tr\left[ \exp\left( \log (F(Z,Y))  -q \log Y\right)\right] 
 &=&
\tr\left[ U_j\exp\left( \log (F(Z,Y))  -q \log 
Y\right)U^*_j\right]\nonumber\\
&=& \tr\left[ \exp\left( \log (F(U_j ZU_j^*,Y))  -q\log 
Y\right)\right]
\end{eqnarray}
Define
$$\widehat{Z} = \frac{1}{|\mathcal{J}|}\sum_{j\in \mathcal{J}}U_j Z U_j^*\ 
,$$
Recall that  $W \mapsto \tr[e^{H + \log W}]$ is concave \cite{L73}.
 Using this, the concavity of $Z \mapsto F(Z,Y)$ 
specified in {\it (1)}, and the monotonicity of the logarithm, averaging 
both sides of (\ref{mainY8}) over $j$ yields
$$
\tr\left[ \exp\left( \log (F(Z,Y))  -q \log Y\right)\right] 
  \leq 
\tr\left[ \exp\left( \log (F(\widehat{Z},Y))  - q\log 
Y\right)\right]  \ .$$
Now making an appropriate choice of $\mathcal{U}$ \cite{D59}, $\widehat{Z} 
$ becomes the ``pinching'' of $Z$ with respect to $Y$; i.e., the orthogonal 
projection in $\mn$ onto the $*$-subalgebra generated by $Y$ and $\one$.  In 
this case, $\widehat{Z} $ and $Y$ commute so that by {\it (3)}, 
$$
 \log (F(\widehat{Z},Y))  -q \log Y = \log\widehat{Z} + 
q\log Y\ .$$
Altogether,
$$\tr\left[ \exp\left( \log (F(Z,Y))  -q \log 
Y\right)\right]   \leq \tr[ \widehat{Z}] = \tr[Z] =1\ ,$$
and this proves  (\ref{mainineqC}).
\end{proof}

For the case $F(X,Y) = Y^{p/2}X^pY^{p/2}$, we can make a similar use of the Peierls-Bogoliubov inequality
but can avoid the appeal to convexity.

\begin{proof}[Proof of Theorem~\ref{clA}]  The inequality we seek to prove is equivalent to
\begin{equation}\label{mainineqBX}
\tr\left[ X\left( \frac1p \log (Y^{p/2}Z^pY^{p/2}) - \log X -  \log 
Y)\right)\right ] \leq 0\ , 
\end{equation}
and again by the Peierls-Bogoliubov inequality it suffices to prove that 
\begin{equation}\label{mainineqCX}
\tr\left[ \exp\left( \frac1p \log (Y^{p/2}Z^pY^{p/2})  - \log 
Y)\right)\right ] \leq 1\ .
\end{equation}
A refined version of the Golden-Thompson inequality due to Friedland and So \cite{FS} says that for all positive $A,B$,  and all $r>0$, 
\begin{equation}\label{FrSo}
\tr[e^{\log A + \log B}] \leq \tr[(A^{r/2}B^r A^{r/2})^{1/r}]\ .
\end{equation}  
and moreover the right hand side is a strictly increasing function of  $r$, unless $A$ and $B$ commute, in which case it is constant in $r$. 
The fact that the right side of \eqref{FrSo} is increasing in $r$ is a conseqence of the Araki-Lieb-Thirring inequality \cite{Ar90}, but here we shall need to know that the increase is strict when $A$ and $B$ do not commute; this is the contribution of \cite{FS}.
Applying \eqref{FrSo} with $r =p$, 
\begin{equation}\label{FrSo2}
\tr\left[ \exp\left( \frac1p \log (Y^{p/2}Z^pY^{p/2})  - \log 
Y)\right)\right ] \leq  \tr[(Y^{-p/2}(Y^{p/2}Z^pY^{p/2}) Y^{-p/2})^{1/p} ] =\tr[Z] = 1\ .
\end{equation}
By the condition for equality in \eqref{FrSo}, there is equality in \eqref{FrSo2} if and only if $(Y^{p/2}Z^pY^{p/2})^{1/p}$ and $Y$ commute, and evidently this is the case if and only if $Z$ and $Y$ commute.

\end{proof}

In the  one parameter family of inequalities provided by  Theorem~\ref{clA}, some are stronger than others. 
It is worth noting that the lower the value of $p>0$ in  (\ref{mainA}) the stronger this inequality is, in the following sense:

\begin{prop}\label{monprop}
The validity of (\ref{mainA}) for $p = p_1$ and for $p = p_2$ implies its validity for $p = p_1+p_2$. 
\end{prop}

\begin{proof} Since there is no constraint 
on $Y$ other than that 
$Y$ is positive, we may replace $Y$ by any power of $Y$. Therefore,  it is 
equivalent to prove  
that for  all $X,Y,Z\in \pn$ such that $\tr[Z] = \tr[X]$ and all $p>0$,
\begin{equation}\label{mainAAA}
\tr[ X \log (YZ^pY))] \leq \tr[X(p\log X +  2\log Y)]\ .
\end{equation}

If (\ref{mainAAA}) is valid for $p = p_1$ and for $p = p_2$, then it is also 
valid for $p = p_1+p_2$:
\begin{eqnarray*}
YZ^{p_1+p_2}Y &=& 
(Y Z^{p_2/2}) Z^{p_1}(Z^{p_2/2}Y)\\
&=& (YZ^{p_2}Y)^{1/2} U^* Z^{p_1}U (YZ^{p_2}Y)^{1/2}\\
&=& (YZ^{p_2}Y)^{1/2} (U^* ZU )^{p_1}(YZ^{p_2}Y)^{1/2}\\
\end{eqnarray*}
where $U (YZ^{p_2}Y)^{1/2}$ is the polar factorization of $Z^{p_2/2}Y$.  
Since $\tr[U^* ZU ] = \tr[Z] = \tr[X]$,
we may apply (\ref{mainAAA}) for $p_1$ to conclude 
$\tr[X\log(Y Z^{p_1+p_2}Y)] \leq 
p_1\tr[X\log X] + \tr[X\log (YZ^{p_2}Y)]$.
One more application of (\ref{mainAAA}), this time with $p=p_2$, yields 
\begin{equation}\label{jb}\tr[X\log(Y Z^{p_1+p_2}Y)] \leq 
(p_1+p_2)\tr[X\log X] + 2\tr[X\log Y]\ .
\end{equation}
By the last line of Corollary~\ref{clA}, the inequality (\ref{jb}) is 
strict 
 if $Z$ and $Y$ do not commute and  at least one of $p_1$ or $p_2$ belongs 
to $(0,1)$. 
\end{proof} 

Our next goal is to prove Theorem~\ref{clB}. As indicated in the Introduction, we will
show that Theorem~\ref{clB} is a consequence of Theorem~\ref{clA}. The 
determination of cases of equality in Theorem~\ref{clA} is essential for 
the proof of the key lemma, which we give now.

\begin{lm}\label{sextend}  Fix $X,Y,Z\in \pn$ such that $\tr[Z] = 
\tr[X]$, and fix  $p> 0$.   Then there is some $\epsilon>0$ so that 
(\ref{mainB}) is valid   for all $s\in [0,\epsilon]$, and such that when $Y$ and $Z$ do not commute, 
(\ref{mainB}) is valid  as a strict for all $s\in (0,\epsilon)$
\end{lm}

\begin{proof} We may suppose, without loss of generality,   that $Y$ and 
$Z$ do not commute since,  if they do commute, the inequality is trivially 
true, just as in Remark~\ref{trcase}.  We compute
\begin{eqnarray*}
\frac{{\rm d}}{{\rm d}s} \tr[ X \log (Y^p\#_s Z^p))]\bigg|_{s=0} &=&
 \tr\left[ X\int_0^\infty \frac{Y^{p/2}}{t+Y^p} \log(Y^{-p/2}ZY^{-p/2}) 
\frac{Y^{p/2}}{t+Y^p} {\rm d}t \right]\\
&=&
 \tr\left[ W \log(Y^{-p/2}ZY^{-p/2})   \right]
\end{eqnarray*}
where
$$W :=  \int_0^\infty \frac{Y^{p/2}}{t+Y^p} X \frac{Y^{p/2}}{t+Y^p} {\rm 
d}t\ .$$
Evidently,
$\tr[W] = \tr[X]= \tr[Z]$. Therefore, by Theorem~\ref{clA} (with $X$ 
replaced by $W$ and $Y$ replaced by $Y^{-1}$), 
$$\tr\left[ W \log(Y^{-p/2}ZY^{-p/2})   \right]  \leq \tr\left[ W( \log W^p
- \log Y^p) \right]\ .$$
Now note that
$$\tr\left[ W \log Y^p\right] = \tr\left[ X \int_0^\infty 
\frac{Y^{p/2}}{t+Y^p} \log Y^p \int_0^\infty \frac{Y^{p/2}}{t+Y^p}{\rm 
d}t\right] = \tr[X \log Y^p]\ .$$
Moreover, by definition $W = \Phi(X)$ where $\Phi$ is a completely 
positive, trace and identity preserving linear map. By 
Lemma~\ref{bosulem} this implies that
$$\tr[W \log W^p] \leq \tr[X \log X^p]\ .$$
 Consequently,
$$
\frac{{\rm d}}{{\rm d}s} \tr[ X \log (Y^p\#_s Z^p) - s \log X^p - 
(1-s)\log Y^p))]\bigg|_{s=0} \leq
\tr[W \log W^p] - \tr[X \log X^p]\ .
$$
Therefore, unless $Y$ and $Z$ commute, the derivative on the left is 
strictly negative, and hence, for some $\epsilon>0$, (\ref{mainB}) is valid as a strict inequality
for all $s\in (0,\epsilon)$. If $Y$ and $Z$ commute, 
(\ref{mainB})
is trivially true for all $p>0$ and all $s\in [0,1]$. 
\end{proof}

\begin{proof} [Proof of Theorem~\ref{clB}]Suppose that (\ref{mainB}) is 
valid for $s= s_1$ and $s= 
s_2$,  Since (by eqs. \eqref{rmet62A} and \eqref{rmet62B} below)
$$(Y^p\#_{s_1} Z^p ) \#_{s_2}Z^p = Y^p\#_{s_1+s_2 -s_1 s_2}Z^p\ ,$$
\begin{eqnarray*}
\tr[X \log (Y^p\#_{s_1+s_2 -s_1 s_2}Z^p)] &=& \tr[ X\log ((Y^p\#_{s_1} Z^p) 
\#_{s_2}Z^p)]\\
&\geq& \tr[ X(s_2 \log X^p + (1-s_2) \log (Y^p\#_{s_1} Z^p))]\\
&\geq&  \tr[ X((s_1+s_2 - s_1s_2) \log X^p+ (1-s_2)(1-s_1) \log Y^p)].
\end{eqnarray*}
Therefore, whenever (\ref{mainB}) is valid for $s= s_1$ and $s= s_2$, it is 
valid for $s= s_1+s_2 - s_1s_2$.

By Lemma~\ref{sextend}, there is some $\epsilon>0$ so that 
(\ref{mainB}) is valid {\em as a strict inequality} for  all $s\in (0,\epsilon)$.  Define an increasing 
sequence $\{t_n\}_{n\in \N}$ recursively by
$t_1 = \epsilon$ and $t_n = 2t_n - t_n^2$ for $n>1$. Then by what we have 
just proved, 
(\ref{mainB}) is valid as a strict inequality for all $s\in (0,t_n)$. Since 
$\lim_{n\to\infty}t_n = 1$, the proof is complete.
\end{proof}

The next goal is to show that the inequality on the right in \eqref{HPpair} 
is a  consequence of Theorem~\ref{clB} by a simple differentiation  argument.  
This simple proof is the new feature, The statement concerning cases of equality was proved in \cite{H94}.

\begin{thm}\label{HPhard} For all $X,Y\in \pn$ and all $p>0$,
\begin{equation}\label{mainB33}
\tr[X(\log X^p + \log Y^p)] \leq \tr[X\log (X^{p/2}Y^pX^{p/2})] \ ,
\end{equation}
and this inequality is strict unless $X$ and $Y$ commute. 
\end{thm}

\begin{proof}  Specializing to the case $Z =X$ in Theorem~\ref{clB}, 
\begin{equation}\label{mainB1}
\tr[ X \log (Y^r\#_s X^r))] \leq \tr[X(s \log X^r +  (1-s)\log Y^r)]
\end{equation}
At $s=1$ both sides of (\ref{mainB1}) equal $\tr[X\log X^r]$, Therefore, we 
may differentiate at $s=1$ to obtain a new inequality. Rearranging terms in (\ref{mainB1}) yields
\begin{equation}\label{mainB2}
\frac{\tr[X\log X^r] - \tr[ X \log (Y^r\#_{s} X^r))]}{1-s}  \geq \tr[X 
(\log X^r - \log Y^r)]\ .
\end{equation}
Taking the limit $s\uparrow 1$ on the left side of (\ref{mainB2}) yields 
${\displaystyle \frac{{\rm d}}{{\rm d}s} \tr[ X \log (Y^r\#_p 
X^r))]\bigg|_{s=1}}$.
>From the integral representation for the logarithm, namely ${\displaystyle 
\log A = \int_0^\infty \left(\frac{1}{\lambda} - 
\frac{1}{\lambda +A}\right) {\rm d}\lambda}$, it follows that for all $A\in 
\pn$ and $H\in \hn$, 
$$\frac{{\rm d}}{{\rm d}u} \log (A+uH) \bigg|_{u=0} =  \int_0^\infty 
\frac{1}{\lambda +A} H \frac{1}{\lambda +A}{\rm d}\lambda\ .$$

Since (see \eqref{rmet62B}) $Y^r\#_{s }X^r = X^s\#_{1-s }Y^s = X^{r/2}(X^{-r/2} Y^r 
X^{-r/2})^{1-p} X^{r/2}$,
$$\frac{{\rm d}}{{\rm d}s} Y^r\#_{s }X^r \big|_{s =1} =  - X^{r/2} 
\log(X^{-r/2}Y^r X^{-r/2})X^{r/2}
=  X^{r/2} \log(X^{r/2}Y^{-r}X^{r/2})X^{r/2}\ ,$$

Altogether, by the cyclicitiy of the trace,
\begin{eqnarray*}
\frac{{\rm d}}{{\rm d}p} \tr[ X \log (Y^r\#_s X^r))] \bigg |_{s=1}  &=&
 \tr\left[ \int_0^\infty \frac{X^{1+r}}{(\lambda +X^r)^2}  {\rm d}\lambda 
\log(X^{r/2}Y^{-r}X^{r/2}) \right]\\
  &=& \tr[X \log(X^{r/2}Y^{-r} X^{r/2})] \ .
\end{eqnarray*}
Replacing $Y$ by $Y^{-1}$ yields (\ref{mainB33}). 

This completes the proof of the inequality itself, and it remains to deal with the cases of equality. Fix $r>0$ and $X$ and $Y$  that do hot commute. By Theorem~\ref{clC} applied with $Z = X$ and $s = 1/2$, there is some $\delta>0$ such that 
\begin{equation}\label{strict}
\tr[X\log (Y\#_{1/2}X)] \leq \tr[X( \tfrac12  \log X + \tfrac12 \log Y)]  - \tfrac12 \delta\ .
\end{equation}
Now use the fact that $Y\#_{3/4}X = (Y\#_{1/2}X)\#_{1/2}X$, and apply (\ref{mainB33}) and then (\ref{strict}):
\begin{eqnarray*}
\tr[X\log (Y\#_{3/4}X)]  &=& \tr[X\log ((Y\#_{1/2}X)\#_{1/2}X)]  \leq \tr[X( \tfrac12  \log X + \tfrac12 \log (Y\#_{1/2}X))]\\
&=& \tfrac12 \tr[X\log X] + \tfrac12\tr[X( Y\#_{1/2}X))]\\
&\leq& \tfrac12 \tr[X\log X] + \tfrac12(\tr[X( \tfrac12  \log X + \tfrac12 \log Y)]  - \tfrac12\delta)\\ 
&=&
\tr[X( \tfrac34  \log X + \tfrac14 \log Y)]  - \tfrac14 \delta\ .
\end{eqnarray*}
We may only apply strict in the last step since $\delta$ depends on $X$ and $Y$, and strict need not hold if $Y$ is replaced by 
$Y\#_{1/2}X$. However, in this case, we may apply (\ref{mainB33}). 

Further iteration of this argument evidently yields the inequalities
$$\tr[X\log (Y\#_{1- t_k}X)] \leq \tr[X( (1-t_k)  \log X + s_k \log Y)]  - t_k\delta\ ,  \qquad  t_k = 2^{-k}\ ,$$
for each $k\in \N$. We may now improve (\ref{mainB2}) to 
\begin{equation}\label{mainB2}
\frac{\tr[X\log X^r] - \tr[ X \log (Y^r\#_{s} X^r))]}{1-s}  \geq \tr[X 
(\log X^r - \log Y^r)] + \delta \ 
\end{equation}
for $s = 1 - 2^{-k}$, $k\in N$. By the calculations above, taking $s\to 1$ along this sequence yields the desired strict inequality. 
\end{proof}

Further inequalities, which we discuss now,
involve an extension of the notion of geometric means. This extension is 
 introduced here and explained in more detail in Appendix C. 
 
Recall that for $t\in [0,1]$ and $X,Y\in \pn$, $X\#_t Y := 
X^{1/2}(X^{-1/2}YX^{-1/2})^t X^{1/2}$. As noted earlier, this formula makes sense for all 
$t\in \R$, and it has a natural geometric meaning. 
The map $t\mapsto X\#_t Y$, defined for $t\in \R$, is a constant speed 
geodesic  running between $X$ and $Y$ for a particular  Riemannian 
metric on the space of positive matrices. 

\begin{defi} \label{gmdef} For $X,Y\in \pn$ and for $t\in \R$, 
\begin{equation}\label{geodes1}
X\#_t Y := X^{1/2}(X^{-1/2}YX^{-1/2})^t X^{1/2}\ .
\end{equation}
\end{defi}
The geometric picture leads to an easy proof of the following identity:
Let $X,Y\in \pn$, and $t_0, t_1\in \R$. Then for all $t\in \R$
\begin{equation}\label{rmet62AAA}
X\#_{(1-t)t_0 + t t_1} Y = (X\#_{t_0}Y)\#_t   (X\#_{t_1}Y)
\end{equation}
See Theorem~\ref{reparam} for the proof.  As a special case, take $t_1 = 0$ 
and $t_0 = 1$. Then, for all $t$,
\begin{equation}\label{rmet62AAA}
X\#_{1-t} Y = Y\#_t X\ .
\end{equation}

With this definition of $X\#_t Y$ for $t\in \R$ we have:

\begin{thm}\label{secA} For all $X,Y,Z\in \pn$ such that $\tr[Z]= \tr[X]$, 
\begin{equation}\label{sec1}
\tr[ X \log (Z^r \#_t Y^r)] \geq \tr[X((1-t)\log X^r + t\log Y^r)]\ .
\end{equation}
is valid for all $t\in [1,\infty)$ and $r>0$. If $Y$ and $Z$ do not commute, the inequality is strict for all $t>1$. 
\end{thm}

The inequalities in Theorem~\ref{secA} and in Theorem~\ref{clB} are 
equivalent.  The following simple identity is the key to this observation:

\begin{lm}\label{ABClemma} For $B,C\in \pn$ and $s\neq 1$, let $A = B\#_sC$. 
Then
\begin{equation}\label{ABClm}
 B = C\#_{1/(1-s)} A    \ .
\end{equation}
\end{lm}

\begin{proof}  Note that by (\ref{geodes1}) and (\ref{rmet62AAA}),  $A = 
B\#_sC$ is equivalent to
$A = C^{1/2}(C^{-1/2}BC^{-1/2})^{1-s} C^{1/2}$, so that 
$C^{-1/2}AC^{-1/2} = (C^{-1/2}B C^{-1/2})^{1-s}$.
\end{proof}

\begin{lm}\label{secequi} Let $X,Y,Z\in \pn$ be such that $\tr[Z]= \tr[X]$. 
Let $r>0$.  
Then (\ref{mainB}) is valid for $s\in (0,1)$  if and only if (\ref{sec1}) 
is valid for $t = 1/(1-s)$. 
\end{lm}

\begin{proof} 
Define $W\in \pn$ by $W^r := Y^r\#_s Z^r$. The identity (\ref{ABClm}) then 
says that 
$Y^r = Z^r\#_{1/(1-s)} W^r$.  Therefore,
\begin{multline}\label{transfo}
\tr[ X \log (Y^r\#_s Z^r)  - s \log X^r -   (1-s)\log Y^r )] =\\
 \tr[X(\log W^r - s\log X^r - (1-s)\log (Z^r\#_{1/(1-s)} W^r)]\ 
.\end{multline} 
Since $s\in (0,1)$, the right side of (\ref{transfo}) is non-positive if and 
only if
$$\tr[X\log(Z^r\#_{1/(1-s)} W^r)] \geq  \tr[X(\tfrac{-s}{1-s}\log X + 
\tfrac{1}{1-s}W^r) ]$$
\end{proof}

With this lemma we can now prove Theorem~\ref{secA}.
\begin{proof}[Proof of Theorem~\ref{secA}]  Lemma~\ref{secequi} 
says that Theorem~\ref{secA} is equivalent to Theorem~\ref{clB}. 
\end{proof}

There is a complement to Theorem~\ref{secA} in the case $Z =X$ that is 
equivalent to a result of Hiai and Petz, who formulate it differently and 
do not discuss extended geometric means.   The statement concerning cases of equality is new. 

\begin{thm}\label{secAA} For all $X,Y\in \pn$, 
\begin{equation}\label{secVYV}
\tr[ X \log (X^r \#_t Y^r)] \geq \tr[X((1-t)\log X^r + t\log Y^r)]\ .
\end{equation}
is valid for all $t\in (-\infty,0]$ and $r>0$. If $Y$ and $X$ do not commute, the inequality is strict for all $t< 0$. 
\end{thm}

\begin{proof}
By  definition \ref{gmdef}
$$
X^r \#_t Y^r  = X^{r/2}(X^{r/2}Y^{-r} X^{r/2})^{|t|}X^{r/2} = X^{r/2} W^r 
X^{r/2} \qquad{\rm where}\qquad 
W := (X^{r/2}Y^{-r} X^{r/2})^{|t|/r}\ .$$ Therefore, by (\ref{mainB33}),
$$\tr[X \log (X^r \#_t Y^r)]  = \tr[X\log (X^{r/2}W^rX^{r/2})] \geq 
r\tr[X\log X] + r\tr[X\log W]\ .$$
By the definition of $W$ and (\ref{rmet62AAA}) once more,
$$\tr[X\log W]  = \frac{|t|}{r} \tr[ X \log((X^{r/2}Y^{-r} X^{r/2})] \geq 
\frac{|t|}{r} r\tr[X(\log X - \log Y)]\ .$$
By combining the inequalities we obtain (\ref{secVYV}).
\end{proof}

The proof given by Hiai and Petz is quite different. It
 uses a tensorization argument.

\section{Quantum Relative Entropy Inequalities}

Theorems~\ref{clC}, \ref{clA} and \ref{clB}
show that the three functions
\begin{equation}\label{varentA}
X,Y \mapsto 
\sup_{Z\in \pn, \tr[Z]= \tr[X]}\left\{\  \tr[ X \log (Y^{-1/2}ZY^{-1/2}))] 
\right\}  +\tr [Y] -\tr [X] \ ,
\end{equation}
\begin{equation}\label{varentB}
X,Y \mapsto 
\sup_{Z\in \pn, \tr[Z]= \tr[X]}\left\{\  \tr[ X \log (Y^{-1}\#_{1/2}Z)^2]  
 \right\}  +\tr [Y] -\tr [X]
\end{equation}
and
\begin{equation}\label{varentC}
X,Y \mapsto 
\sup_{Z\in \pn, \tr[Z]= \tr[X]}\left\{\  \tr\left[ 
X \log \left( \int_0^\infty \frac{1}{\lambda +Y}Z\frac{1}{\lambda +Y}{\rm 
d}\lambda
  \right)
 \right] \right\}  +\tr [Y] -\tr [X]
\end{equation}
are all bounded above by the Umegaki relative entropy $X,Y \mapsto 
\tr[X(\log X - \log Y)]+\tr [Y] -\tr [X]  $.  
The next lemma shows that these 
functions are actually one and the same.

\begin{prop}\label{allsame}  The three functions 
defined in 
(\ref{varentA}),  (\ref{varentB}) and 
(\ref{varentC}) are all equal to to the Donald relative  entropy
$ D_D(X||Y) $.   Consequently, for all $X,Y\in 
\pn$,
\begin{equation}\label{vnent3}
D_D(X||Y) \leq  
D(X||Y)\ .
\end{equation}
\end{prop}

\begin{proof}[Proof of Proposition~\ref{allsame}]  The first thing to notice is 
that the relaxed constraint $\tr[YQ] \leq \tr[X]$ imposes the same 
restriction in (\ref{geom31}) as does the hard constraint  $\tr[YQ] = 
\tr[X]$ since, if  $\tr[YQ] < \tr[X]$, we may replace $Q$ by
$(\tr[X]/\tr[YQ])Q$ so that the hard constraint is satisfied.  Thus we may 
replace the relaxed constraint in (\ref{geom31}) by the hard constraint 
without affecting the function $D_D(X||Y)$.  This will be convenient in the 
lemma, though elsewhere the relaxed constraint will be essential. 

Next, for each of (\ref{varentA}),  (\ref{varentB}) and 
(\ref{varentC}) we make a change of variables. In the first case, define 
$\Phi: \pn \to \pn$ by
$\Phi(Z) = Y^{-1/2}ZY^{-1/2} := Q$. Then $\Phi$ is invertible with 
$\Phi^{-1}(Q) = Y^{1/2}QY^{1/2}$. Under this change of variables, the 
constraint $\tr[X] = \tr[Z]$ becomes.
$\tr[X] = \tr[Y^{1/2}QY^{1/2}] = \tr[YQ]\ .$  Thus (\ref{varentA}) gives us 
another expression for the Donald relative entropy. 

For the function in (\ref{varentB}), we make a similar change of variables. 
Define $\Phi:\pn \to \pn$ by
 $\Phi(Z) = Z\#_{1/2}Y := Q^{1/2}$ from $\pn$ to $\pn$.  This map is 
invertible: 
It follows by direct computation from the definition (\ref{geomean}) that  
for 
$Q^{1/2} := Z\#_{1/2}Y^{-1}$, 
$Z = Q^{1/2}YQ^{1/2}$, so that $\Phi^{-1}(Q) =  Q^{1/2}YQ^{1/2}$. (This has 
an interesting and useful geometric interpretation that is discussed in 
Appendix C.)  
Under this change of variables, the constraint $\tr[X] = \tr[Z]$ becomes.
$\tr[X] = \tr[Q^{1/2}YQ^{1/2}] = \tr[YQ]$ . Thus (\ref{varentB}) gives 
another expression for the Donald relative entropy. 

Finally, for the function in (\ref{varentC}), we make a similar change of 
variables. Define $\Phi:\pn \to \pn$ by
 $$\Phi(Z) = \int_0^\infty \frac{1}{\lambda +Y}Z\frac{1}{\lambda +Y}{\rm 
d}\lambda := Q^{1/2}$$ 
 from $\pn$ to $\pn$.  This map is invertible: 
 ${\displaystyle \Phi^{-1}(Q) = \int_0^1 Y^{1-s} Q Y^{s}{\rm d}s}$.
 Under this change of variables, the constraint $\tr[X] = \tr[Z]$ becomes
${\displaystyle \tr[X] = \tr[\int_0^1 Y^{1-s} Q Y^{s}{\rm d}s] = \tr[YQ]}$.
\end{proof}

With the Donald relative entropy  having taken center stage, we now 
bend 
our efforts to establishing some of its properties.

\begin{lm}\label{uniqueQ} Fix $X,Y\in \pn$, and define $\mathcal{K}_{X,Y} := 
\{Q\in \overline{\pn}\ :\ \tr[QY] \leq \tr[X]\}$.
There exists a unique $Q_{X,Y} \in \mathcal{K}_{X,Y}$   such that 
$\tr[Q_{X,Y}Y] \leq \tr[X]$
and such that
$$\tr[X \log Q_{X,Y}] > \tr[X\log Q]$$
for all other $Q\in \mathcal{K}_{X,Y}$. The equation 
\begin{equation}\label{Qeq}
\int_0^\infty \frac{1}{t +Q} X \frac{1}{t +Q}{\rm d}t = Y\ .
\end{equation}
has a unique solution in $\pn$, and this unique solution is the unique 
maximizer $Q_{X,Y}$.
\end{lm}

\begin{proof} Note that $\mathcal{K}_{X,Y}$ is a compact, convex set. 
Since $Q\mapsto \log Q$ is strictly concave, $Q\mapsto \tr[X\log Q]$ is 
strictly concave on 
$\mathcal{K}_{X,Y}$, and it has the value $-\infty$ on $\partial{\pn}\cap 
\mathcal{K}_{X,Y}$, there is a unique maximizer
$Q_{X,Y} $ that lies in $\pn\cap \mathcal{K}_{X,Y}$.

Let $H\in \hn$ be such that $\tr[HY] = 0$. For all $t$ in a neighborhood of 
$0$, $Q_{X,Y} + tH \in 
\pn\cap \mathcal{K}_{X,Y}$. Differentiating in $t$ at $t=0$ yields
$$0 =  \int_0^\infty \left(\tr\left[X\frac{1}{t +Q_{X,Y}} H \frac{1}{t 
+Q_{X,Y}}\right] \right){\rm d}t 
= \tr\left[ H\left( \int_0^\infty \frac{1}{t +Q_{X,Y}} X \frac{1}{t 
+Q_{X,Y}}{\rm d}t  \right) \right]\ ,$$
and hence
$$\int_0^\infty \frac{1}{t +Q_{X,Y}} X \frac{1}{t +Q_{X,Y}}{\rm d}t   = 
\lambda Y$$
for some $\lambda\in \R$.   Multiplying through on both sides by 
$Q_{X,Y}^{1/2}$ and taking the trace yields $\lambda = 1$, which shows that 
$Q_{X,Y}$ solves (\ref{Qeq}). Conversely, any solution of (\ref{Qeq}) 
yields a critical point of our strictly concave functional, and hence must 
be the unique maximizer. 
\end{proof}

\begin{remark} \label{qcommute} There is one special case for which we can 
give a formula for the solution $Q_{X,Y}$ to \eqref{Qeq}: When $X$ and $Y$ 
commute,
$Q_{X,Y} = XY^{-1}$.
\end{remark}

\begin{lm}\label{scalinglem}  For all $X,Y\in \pn$ and 
all $\lambda>0$,  
\begin{equation}\label{scaling1}
D_D(\lambda X,\lambda Y) = \lambda D_D(X,Y)\ ,
\end{equation}
and
\begin{equation}\label{scaling2}
D_D(\lambda X, Y) = \lambda D_D(X,Y) + \lambda\log \lambda\tr[X] + 
(1-\lambda)\tr[Y]\ ,
\end{equation}
\end{lm}

\begin{proof} By (\ref{Qeq}) the maximizer $Q_{X,Y}$ in Lemma~\ref{uniqueQ} 
satisfies the scaling relations
\begin{equation}\label{scaling3}
Q_{\lambda X,Y} = \lambda Q_{X,Y} \qquad{\rm and}\qquad Q_{X,\lambda Y} = 
\lambda^{-1}Q_{X,Y}\ ,
\end{equation}
and (\ref{scaling1}) follows immediately. 
Next, by (\ref{scaling3}) again,
$$D_D(\lambda X||Y)   = \lambda\left(\tr[X\log Q_{X,Y}] + \tr[Y] - 
\tr[X]\right) + \lambda\log \lambda\tr[X] + (1-\lambda)\tr[Y]\ ,$$
which proves (\ref{scaling2}).
\end{proof}

\begin{lm}\label{copo} If $X$ and $Y$ commute,
$
D_D(X||Y) = D(X||Y)\ .
$
\end{lm}

\begin{proof}  Let $\{U_1,\dots,U_N\}$ be any set of  unitary matrices that 
commute with $X$ and $Y$. Then 
for each $j=1,\dots,n$, $\tr[Y(U_j^*QU_j)] = \tr[YQ]$. Define 
$$\widehat{Q} = \frac{1}{N}\sum_{j=1}^NU_j^*QU_j\ .$$
For an appropriate choice of the set $\{U_1,\dots,U_N\}$, $\widehat{Q}$ is 
the orthogonal projection of $Q$, with respect to the Hilbert-Schmidt inner 
product,  onto the abelian subalgebra of $\mn$ generated by $X$, $Y$ and 
$\one$ \cite{D59}. By the  concavity of the logarithm,
$$
\tr[X\log\widehat{Q}] \geq \frac{1}{N} \sum_{j=1}^N \tr [X \log 
(U^*_jQU_j)] = \frac{1}{N} \sum_{j=1}^N \tr [UXU^* \log Q]  = \tr[X\log Q] \ 
. $$
Therefore, in taking the supremum, we need only consider operators $Q$ that 
commute with both $X$ and $Y$. 
The claim now follows by Remark \ref{qcommute}.
\end{proof}

\begin{remark} Another simple proof of this can be given using Donald's original formula \eqref{dondef}.
\end{remark}

We have now proved that $D_D$ has properties (2) and (3) in the 
Definition \ref{qred} of  relative entropy, and have already observed that it inherits joint convexity from the Umegaki relative entropy though its original definition by Donald. 

We now compute the partial Legendre transform of $D_D(X||Y)$. {\em In doing so
we arrive at a direct proof of the joint 
convexity of $D_D(X||Y)$, independent of the joint convexity of the Umegaki relative entropy.}  We first prove Lemma~\ref{logconc}.

\begin{proof}[Proof of Lemma~\ref{logconc}] For $X\in \pn$, define $a = \tr[X]$ and $W := a^{-1}X$, so 
that $W$ is a density matrix.
Then
$$\tr[XH]  - R(X||Y)  =  a\tr[WH] - aR(W||Y) - a\log a  - (1-a)\tr[Y]\ ,$$
Therefore,
\begin{eqnarray*}
\Psi_R(H,Y) &=& \sup_{a>0}\left\{ a \sup_{W\in \pn}\left\{ \tr[WH] - 
D_R(W||Y)\ :\ \tr[W] =1\ \right\}   +a \tr[Y]  - a \log a \right\} - 
\tr[Y]\\
&=& \sup_{a>0}\left\{  a(\Phi_R(H,Y) + \tr[Y])  - a \log a\right\} - \tr[Y]\ 
.
\end{eqnarray*}
Now use the fact that for all $a>0$ and all $b\in \R$, $a\log a + e^{b-1} 
\geq ab$ with equality if and only if $b = 1+ \log a$ to conclude that 
(\ref{dualrel}) is valid.  \end{proof}

The function $D_D$ evidently satisfies the conditions of this lemma. Our immediate
goal is to compute $\Phi_R(H,Y)$ for this choice of $R$, and to show   its
concavity as a function of $Y$. Recall the definition
\begin{equation}\label{geom33}
\Phi_D(H,Y) :=  \sup_{X>0, \tr[X] =1}  \{  \tr[XH] - D_D(X||Y) \ \}\ .
\end{equation}
We wish to evaluate the supremum as explicitly as possible. 

\begin{lm}\label{mmlem} For $H\in \hn$ and $Y\in \pn$,
\begin{equation}\label{MM5}
\Phi_D(H,Y)\  =\   1 -  \tr[Y]  +   \inf\big\{  \lambda_{\max}\left( H - \log Q)\right)  
\ :\ Q\in \pn\ , \tr[QY] \leq 1\ \big\}\ 
\end{equation}
where for any self-adjoint operator $K$, $\lambda_{\max}(K)$ is the largest 
eigenvalue of $K$. 
\end{lm}

Our proof of (\ref{MM5}) makes use of  a Minimax Theorem;  such theorems  give conditions 
under which a function $f(x,y)$ on $A\times B$ satisfies
\begin{equation}\label{MM}
\sup_{x \in  A} \inf_{\phantom{.}y\in B\phantom{\dot I}}  f(x,y) = 
\inf_{\phantom{.}y\in B\phantom{\dot I}}\sup_{x\in A} f(x,y)\ .
\end{equation}
The original Minimax Theorem was proved by von Neumann \cite{vN28}. While most 
of his paper deals with the case in which $f$ is a bilinear function on 
$\R^m\times \R^n$ for some $m$ and $n$, and $A$ and $B$ are simplexes, he 
also proves \cite[p. 309]{vN28} a more general results for functions on 
$\R\times \R$ that are quasi-concave in $X$ and quasi convex in $y$. 
According to Kuhn and Tucker \cite[p. 113]{KT58}, a multidimensional version 
of this is implicit in the paper. von Neumann's work inspired host of 
researchers to undertake extensions and generalizations; \cite{FKK} contains 
a useful survey. A theorem of Peck and Dulmage \cite{PD} serves our purpose. 
 See \cite{S58} for a more general extension. 

\begin{thm}[Peck and Dulmage]\label{PDMM}  Let $\mathcal{X}$ be a 
topological vector space, and let  $\mathcal{Y}$ be a  vector space. Let 
$A\subset \mathcal{X}$  be non-empty compact and convex, and let $B \subset 
\mathcal{Y}$  be non-empty and
convex.  Let $f$ be a real valued  
function on $A\times B$
such that for each fixed $y\in B$, $x\mapsto f(x,y)$ is concave and upper 
semicontinuous, and for each fixed $x\in A$, 
$y\mapsto f(x,y)$ is convex. Then (\ref{MM}) is valid. 
\end{thm}

\begin{proof}[Proof of Lemma~\ref{mmlem}] The formula (\ref{MM5}) has been proved above. 

Define $\mathcal{X}= \mathcal{Y} = \mn$, $A = \{  W\in \overline{\pn}\ :\  
\tr[W] =1\ \}$
and $B := \{W \in \pn\ :\ \tr[WY]  \leq 1\}$.  For $H\in \hn$, define
$$f(X,Q) := \tr[X(H - \log Q)]\ .$$
Then the hypotheses of Theorem~\ref{PDMM} are satisfied, and  hence
\begin{equation}\label{MM1}
\sup_{X \in  A} \inf_{\phantom{.}Q\in B\phantom{\dot I}}  f(X,Q) = 
\inf_{\phantom{.}Q\in B\phantom{\dot I}}\sup_{X\in A} f(X,Q)\ .
\end{equation}

Using the definition (\ref{geom33}) and the identity (\ref{MM1})
\begin{eqnarray}\label{lmax}
 \Phi_D(H,Y) + \tr[Y] -1&:=&  \sup_{X>0, \tr[X] =1}  \left\{  \tr[XH] - 
\sup_{Q>0, \tr[QY] \leq1}\{  \tr[X\log Q]\  \} \ \right\} \nonumber \\
&:=&  \sup_{X>0, \tr[X] =1} \inf_{Q>0, \tr[QY] \leq1} \left\{  
\tr\left[X\left( H -   \log Q\right)\right]\   \ \right\} \nonumber \\
&=&   \inf_{Q>0, \tr[QY] \leq1} \sup_{X>0, \tr[X] =1} \left\{  
\tr\left[X\left( H - \log Q)\right)\right]\   \ \right\}  \nonumber \\
&=&   \inf_{Q>0, \tr[QY] \leq 1} \lambda_{\max}\left( H - \log Q)\right)  
\end{eqnarray}
\end{proof}

\begin{lm}\label{jcne} For each $H\in \hn$, $Y\mapsto \Phi_D(H,Y)$ is 
concave.  \end{lm}

\begin{proof}  Fix $Y>0$ and let $A\in \hn$ be such that $Y_\pm := Y\pm A$ 
are both positive. Let $Q$ be optimal in the variational formula (\ref{MM5}) 
for $\Phi(H,Y)$.  We claim that there exists $c\in \R$ so that 
\begin{equation}\label{geom51}
\tr[Y_+Qe^c] \leq 1 \qquad{\rm and}\qquad \tr[Y_-Qe^{-c}] \leq 1\ .
\end{equation}
Suppose for the moment that this is true. Then
$$\lambda_{\max}(H - \log Q) = \frac12  \lambda_{\max}(H - \log (Qe^c) + 
\frac12  \lambda_{\max}(H - \log (Qe^{-c})\ .$$
By (\ref{geom51}),
$$\Phi_D(H,Y) \geq \frac12 \Phi_D(H,Y_+) + \frac12 \Phi_D(H,Y)\ ,$$
which proves midpoint concavity. The general concavity statement follows by 
continuity. 

To complete this part of the proof, it remains to show that we can choose 
$c\in \R$ so that (\ref{geom51}) is satisfied.    Define $a := \tr[Q A]$.  
Since $Y\pm A > 0$, and $\tr[Q(Y\pm A)] > 0$, which is the same as
$1\pm a > 0$. That is,  $ |a| < 1$.  We then compute
$$\tr[Y_+Qe^{c}] = e^c \tr[YQ + AQ]  = e^c(1+a)$$
and likewise, $\tr[Y_-Qe^{-c}]  - e^{-c}(1-a)$.  We wish to choose $c$ so 
that
$$e^c(1+a) \leq 1 \quad{\rm and}\quad e^{-c} (1-a)\leq 1\ .$$
This is the same as
$$\log(1-a) \leq   c \leq -\log(1+a)\ .$$
Since $-\log(1+a) - \log(1-a) = -\log(1-a^2) > 0$. the interval $[\log(1-a) 
, -\log(1+a)]$ is non-empty, and we may choose any $c$ in this interval.   
\end{proof}

We may now improve on  Lemma~\ref{jcne}: Not only is $\Phi_D(H,Y)$ concave in $Y$; {\em its exponential is also concave in $Y$}.

\begin{thm}\label{expconcavbe} For all $H\in \hn$, the function
\begin{equation}\label{legscale2}
Y \mapsto \exp\left( \inf_{Q>0, \tr[QY] \leq 1} \lambda_{\max}\left( H - \log Q)\right)\right)  
\end{equation}
is concave on $\pn$.  Moreover,  for all  $H,K\in \hn$, 
\begin{equation}\label{geom52R}
\log(\tr[e^{H+K}] )\leq \  \inf_{Q>0, \tr[Qe^K] \leq 1} 
\lambda_{\max}\left( H - \log Q)\right)  \  \leq  \log( \tr[e^{H} e^{K}]).
\end{equation}
These inequalities  improve upon 
the Golden-Thompson inequality.
\end{thm}

\begin{proof}  Let $\Psi_D(H,Y)$ be the partial Legendre transform of $D(X||Y)$ in $X$ without any restriction on
$X$:  
\begin{equation}\label{geom33R}
\Psi_D(H,Y) :=  \sup_{X>0}  \{  \tr[XH] - D_D(X||Y) \ \}\ ,
\end{equation}
By \cite[Theorem 1.1]{CL08}, and the joint convexity of $D_D(X||Y)$,  $\Psi_D(H,Y)$ is concave in $Y$ for each
fixed $H \in \hn$. 
By Lemma~\ref{logconc}, 
$$\Psi_D(H,Y) = e^{ \Phi_D(X,Y) + \tr[Y] -1}  - \tr[Y]\ ,$$
and thus we conclude
\begin{equation}\label{ident}\Psi_D(H,Y)  =  
\exp\left( \inf_{Q>0, \tr[QY] \leq 1} \lambda_{\max}\left( H - \log Q)\right)\right) - \tr[Y]\ . 
\end{equation}
The inequality $\Psi(H,Y) \leq \Psi_D(H,Y)$ follows from $D_D(X||Y) 
\leq D(X||Y)$ and the order reversing property of Legendre transforms. 
Taking exponentials and writing $Y = e^K$  yields the first inequality in 
(\ref{geom52R}).  Finally, choosing
${\displaystyle Q := \frac{e^{H}}{\tr[e^H Y]}}$
so that the constraint $\tr[QY] \leq 1$ is satisfied, we obtain
$\Phi_D(H,Y) \leq \log(\tr[e^H Y])$.
Taking exponentials and writing $Y = e^K$  now yields the second inequality 
in (\ref{geom52R}).
\end{proof}

The proof that the function in \eqref{legscale2} is concave has two components. 
One is the identification \eqref{ident} of this function with $\Psi_D(H,Y)$. The second 
makes use of the direct analog of an argument of Tropp \cite{Tr12} proving the concavity in 
$Y$  of $\tr[e^{H + \log Y}] = \Psi(H,Y) + \tr[Y]$ as a consequence of the joint convexity of the 
Umegaki relative entropy. Once one has the formula \eqref{ident}, the convexity of the 
function in \eqref{legscale2} follows from the same argument, applied instead to the Donald 
relative entropy, which is also jointly convex. 

However, it is of interest to note here that this argument can be run in reverse to deduce 
the joint convexity of the Donald relative entropy without invoking the joint convexity 
of the Umegaki relative entropy. 
To see this, note that Lemma~\ref{jcne} provides a simple direct proof of  the concavity in $Y$ of 
$\Phi_D(H,Y)$.  By the Fenchel-Moreau Theorem,
for all density matrices $X$
\begin{equation}\label{PhiRdD}
D_D(X||Y)   = \sup_{H\in \hn}\{ \tr[XH] - \Phi_D(H,Y)  \} \ .
\end{equation}
For each fixed $H\in \hn$, $X,Y \mapsto \tr[XH] - \Phi_R(H,Y)$ is evidently jointly convex. 
Since the supremum of any family of  convex functions is convex, we conclude that   
with the $X$ variable restricted to be a density matrix,  $X,Y \mapsto D_D(X||Y)$ is jointly convex. 
The restriction on $X$ is then easily removed; see Lemma~\ref{homoge} below. 
This gives an elementary proof of the joint convexity of $D_D(X||Y)$. 

It is somewhat surprising the the joint convexity of the Umegaki relative entropy is deeper than the joint convexity of 
either $D_D(X||Y)$ or $D_{BS}(X||Y)$. In fact, the simple proof by Fujii and Kamei that the latter is 
jointly convex stems from a joint {\em operator convexity} result; see the discussion in Appendix C. 
The joint convexity of the Umegaki relative entropy, in contrast,
stems from  the basic concavity theorem in \cite{L73}.

\begin{lm}\label{homoge}  Let $f(x,y)$ be a $(-\infty, \infty]$ valued function on 
$\R^m\times \R^n$ that is homogeneous of degree one. Let
$a\in \R^m$, and let $K_a = \{ x\in \R^m\ :\ \langle a,x\rangle =1\}$, 
and suppose that whenever $f(x,y) < \infty$, $\langle a,x\rangle > 0$. 
If $f$ is convex on $K_a\times \R^n$, then it is convex on $\R^m\times \R^n$. 
\end{lm}

\begin{proof} Let $x_1,x_2\in \R^n$ and $y_1,y_2\in \R^n$. We may suppose that $f(x_1,y_1),f(x_2,y_2) < \infty$. Define
$\alpha_1 = \langle a,x_1\rangle$ and $\alpha_2 = \langle a,x_2\rangle$. Than  $\alpha_1,\alpha_2> 0$, 
and $u_1/\alpha_1,x_2/\alpha_2\in K_a$. With $\lambda := \alpha_1/(\alpha_1+\alpha_2)$, 
\begin{eqnarray*}
f(x_1+x_2,y_1+y_2) &=&(\alpha_1+\alpha_2)
 f\left (\lambda \frac{x_1}{\alpha_1} +(1-\lambda) \frac{x_2}{\alpha_2} ,\lambda \frac{y_1}{\alpha_1} 
 +(1-\lambda)\frac{y_2}{\alpha_2}\right )\\
 &\leq&
 (\alpha_1+\alpha_2)\lambda  f\left ( \frac{x_1}{\alpha_1} , \frac{y_1}{\alpha_1} 
 \right ) +  (\alpha_1+\alpha_2)(1-\lambda)  f\left ( \frac{x_2}{\alpha_2} , \frac{y_2}{\alpha_2} 
 \right )\\
 &=&f(x_1,y_1) + f(x_2,y_2)\ .
 \end{eqnarray*}
Thus, $f$ is subaddtive on $\R^m\times \R^m$, and by the homogeneity once more, jointly convex. 
\end{proof}

We next provide the proof of Proposition~\ref{pinsker}, which we recall says that any quantum relative entropy functional satisfies the inequality
\begin{equation}\label{pinsk22}
R(X||W) \geq \tfrac12 \tr[X]  \left\Vert \frac{X}{\tr[X]} 
-\frac{W}{\tr[W]}\right\Vert_1^2
\end{equation}
for all $X,W\in \pn$, 
where $\|\cdot \|_1$ denotes the trace norm.

\begin{proof}[Proof of Proposition~\ref{pinsker}]  By scaling, it suffices to show that when $X$ and $W$ are density matrices, 
\begin{equation}\label{pinsk222}
R(X||W) \geq \tfrac12 \ \left\Vert X 
- W\right\Vert_1^2
\end{equation}
Let $X$ and $W$ be density matrices and define $H = X - W$. Let $P$ be the spectral 
projection onto the subspace of $\C^n$ spanned be the eigenvectors of $H$ with 
non-negative eigenvalues. Let $\mathcal{A}$ be the $*$-subalgebra of $\mn$ generated by 
$H$ and $\one$, and let ${\rm E}_{\mathcal{A}}$ be the orthogonal projection in $\mn$ 
equipped with the Hilbert-Schmidt inner product onto $\mathcal{A}$. Then $A \mapsto 
{\rm E}_{\mathcal{A}}A$ is a convex operation \cite{D59}, and then by the joint convexity of $R$, 
\begin{equation}\label{pinsk223}
R(X||Y) \geq R({\rm E}_{\mathcal{A}}X || {\rm E}_{\mathcal{A}}Y)\ .
\end{equation}
Since both ${\rm E}_{\mathcal{A}}X$ and ${\rm E}_{\mathcal{A}}Y$ belong to the commutative 
algebra $\mathcal{A}$, \eqref{pinsk223} together with property {\it (3)} in the definition of quantum relative entropies then gives us
$$R(X||Y) \geq D({\rm E}_{\mathcal{A}}X || {\rm E}_{\mathcal{A}}Y)\ .$$
Since $\|{\rm E}_{\mathcal{A}}X - {\rm E}_{\mathcal{A}}Y\|_1 = \|X - Y\|_1$, 
the inequality now follows from the {\em classical} Csiszar-Kullback-Leibler-Pinsker inequality \cite{C67,KL51,K67,Pi64} on a two-point probability space.
\end{proof}

\begin{remark}  The proof of the lower bound (\ref{pinsk222}) given here is essentially the same as the proof for the case of the Umegaki relative entropy given in \cite{HOT81}. The proof gives one reason for attaching importance to the joint convexity property, and since it is short, we spelled it out to emphasize this. 
\end{remark} 

We conclude this section with a brief discussion of the failure of 
convexity of the function   $\phi(X,Y) = \tr X^{1/2}\log (Y^{-1/2} X 
Y^{-1/2})  X^{1/2}$.
We recall that if we write this in the other order, i.e., define the 
function  $\psi(X,Y) = \tr X^{1/2}\log (X^{1/2} Y^{-1} X^{1/2}) X^{1/2}$, 
the function $\psi$  is jointly convex. In fact $\psi$ is {\em operator 
convex} if the trace is omitted. We might have hoped, therefore, that $\phi$ 
would at least be convex in $Y$ alone, and even have hoped that 
$\log (Y^{-1/2} X Y^{-1/2})$ is operator convex in $Y$. Neither of these 
things is true. The following lemma precludes the operator convexity. 

\begin{lm} Let $F$ be a function mapping the set of positive semidefinite 
matrices into itself. Let $f: [0,\infty)\to \R$ be a concave, monotone 
increasing function. If $Y \mapsto f(F(Y))$ is operator convex, then $Y 
\mapsto F(Y)$ is operator convex.
\end{lm}

\begin{proof} If $Y \mapsto F(Y)$ is not operator convex, then there is a 
unit vector $v$ and there are density matrices $Y_1$ and $Y_2$ such that 
with $Y = \frac12(Y_1+Y_2)$, 
$$\langle v, F(Y) v\rangle < \frac12\left(\langle v, F(Y_1)v\rangle + 
\langle v, F(Y_2)v\rangle\right)\ .$$
By Jensen's inequality, for all density matrices $X$,
$\langle v, f(F(X))v\rangle  \leq  f(\langle v, F(X)v\rangle)$.
Therefore,
\begin{eqnarray*}
\frac12\left( \langle v, f(F(Y_1))v\rangle + \langle v, 
f(F(Y_2))v\rangle\right) &\leq&
\frac12\left( f(\langle v, F(Y_1)v\rangle) + f(\langle v, 
F(Y_2)v\rangle)\right)\\
&\leq&
f\left(\frac12\left( \langle v, F(Y_1)v\rangle + \langle v, 
F(Y_2)v\rangle\right)\right) < \langle v, f(F(Y))v\rangle\ .
\end{eqnarray*}
\end{proof}

By the lemma, if $Y \mapsto \log(Y^{-1/2}ZY^{-1/2})$ were convex, $Y\mapsto 
Y^{-1/2}ZY^{-1/2}$ would be convex. But this may be shown to be false in the 
$2\times 2$ case by simple computations in an neighborhood of the identity 
with $Z$ a rank-one projector.  A more intricate computation of the same 
type shows that -- even with the trace -- convexity fails.

\section{Exponential Inequalities Related to the\\
Golden Thompson 
Inequality}

Let $\Psi(H,Y)$ be given in \eqref{psiu} and  $\Psi_D(H,Y)$ be given in 
\eqref{geom33R}.
We have seen in the previous section that the inequality $D_D(X||Y)\leq 
D(X||Y)$ leads to the inequality
$\Psi(H,Y) \leq \Psi_D(H,Y)$. This inequality, which may be 
written 
explicitly as
\begin{equation}\label{gtbet}
\tr[e^{H + \log Y}] \leq \exp\left( \inf\{  \lambda_{\max}\left( H - \log 
Q)\right)  \ :\ Q\in \pn\ , \tr[QY] \leq 1\ \} \right) \ ,
\end{equation}
immediately implies the Golden-Thompson inequality through the simple 
choice 
$Q = e^H/\tr[Y e^H]$. 
The $Q$ chosen here is optimal only when $H$ and $Y$ commute. Otherwise,
there is a better choice for $Q$, which will lead to a tighter upper bound.

A similar analysis can be made with respect to the BS relative entropy.
Define $\Psi_{BS}(H,Y)$ by
\begin{equation}\label{psibs}
\Psi_{BS}(H,Y) := \sup\{ \tr[HX] - D_{BS}(X||Y)\  : X\in \pn \}\ .
\end{equation}
The inequality $D(X||Y) \leq D_{BS}(X||Y)$ together with 
Lemma~\ref{logconc}  gives
\begin{equation}\label{psibs2}
\Psi_{BS}(H,Y)   \leq \Psi(H,Y) = \tr[e^{H +\log Y}] - \tr[Y] \  .
\end{equation}
It does not seem possible to compute $\Psi_{BS}(H,Y)$ explicitly, 
but it is possible to give an alternate expression for it in terms of the 
solutions of a non-linear matrix 
equation similar to the one (\ref{Qeq})  that arises in the context of the 
Donald relative entropy.

Writing out the identity $X\#_t Y = Y\#_{1-t}X$ gives
$$X^{1/2}(X^{-1/2}Y X^{-1/2})^t X^{1/2}   =  Y^{1/2}(Y^{-1/2} X 
Y^{-1/2})^{1-t} Y^{1/2} \ .$$
Differentiating at $t=0$ yields
$$X^{1/2}\log (X^{1/2}Y^{-1} X^{1/2}) X^{1/2} = Y^{1/2}(Y^{-1/2} X Y^{-1/2}) 
\log (Y^{-1/2} X Y^{-1/2})   Y^{1/2}\ .$$
This provides an alternate expression for $D_{BS}(X||Y)$ that involves $X$ 
in a somewhat simpler way that is advantageous for the partial Legendre 
transform in $X$:
\begin{equation}\label{BSalt}
D_{BS}(X||Y) = \tr[ Y f (Y^{-1/2} X Y^{-1/2}) ] - \tr[X] + \tr[Y]
\end{equation}
where $f(x) = x \log x$.  A different derivation of this formula may be 
found in \cite{HP93}.

Introducing the variable $R = Y^{-1/2} X Y^{-1/2}$ we have, for all $H\in 
\hn$,
\begin{eqnarray*}
\tr[XH] - D_{BS}(X||Y)  &=&  \tr[X(H+ \one)]  - \tr  \tr[ Y f (Y^{-1/2} X 
Y^{-1/2}) ]  - \tr[Y]\\
&=& \tr[R(Y^{1/2}(H+\one)Y^{1/2})] - \tr[Y f(R)] - \tr[Y] \ .
\end{eqnarray*}
Therefore,
\begin{equation}\label{psibsalt}
\Psi_{BS}(H,Y) + \tr[Y]  = \sup_{R\in \pn}\left\{  
\tr[R(Y^{1/2}(H+\one)Y^{1/2})] - \tr[Y f(R)]  \right\} \ .
\end{equation}

When $Y$ and $H$ commute, the supremum on the right is achieved at $R = e^H$ 
since for this choice of $R$,
$$ \tr[R(Y^{1/2}(H+\one)Y^{1/2})] - \tr[Y f(R)]  = \tr[Ye^H] = \tr[e^{H + 
\log Y}]\ ,$$
and by (\ref{psibs2}), this is the maximum possible value. 

In general, without assuming that $H$ and $Y$ commute, this choice of $R$ 
and (\ref{psibs2}) yields an interesting inequality.
\begin{thm}
 For all self-adjoint $H$ and $L$, 
 \begin{equation}
\tr[e^He^L] - \tr[e^{H+L}]  \leq \tr[e^HHe^L] - \tr[e^H e^{L/2}H e^{L/2}]\ .
 \end{equation}
\end{thm}

\begin{proof}
With the choice $R = e^H$, the inequality \eqref{psibs2} together 
with \eqref{psibsalt} yields  
$$
\tr[e^H (Y^{1/2}H Y^{1/2} + Y)]  - \tr[ Y e^H H] \leq \tr[e^{H + \log Y}]\ ,
$$
or, rearranging terms,
$$
\tr[e^HY]  - \tr[e^{H + \log 
Y}] \leq 
  \tr[e^HH Y]  - \tr[ e^H(Y^{1/2}HY^{1/2})].
$$
The inequality is proved
by writing $Y = e^L$.
\end{proof}

We now turn to the specification of the actual maximizer. 

\begin{lm}\label{BSmax} For $K\in \hn$ and $Y\in \pn$, the function
$$R \mapsto \tr[RK] - \tr[Y f(R)] $$
 on $\overline{\pn}$ has a unique maximizer $R_{K,Y}$  in $\overline{\pn}$ 
which is contained in 
 $\pn$, and $R_{K,Y}$  is the unique critical point of this function in 
$\pn$. 
\end{lm}

\begin{proof} Since $f$ is strictly operator convex, $R \mapsto \tr[RK] - 
\tr[Y f(R)]$ is strictly concave. There are no local maximizers on the 
boundary on $\pn$ since $\lim_{x\downarrow 0}(-f'(x)) = \infty$, so that if 
$R$ has a zero eigenvalue, 
a small perturbation of $R$ will yield a higher value. 

Finally, 
$$ \tr[RK] - \tr[Y f(R)]   \leq \|K\| \tr[ R - \tfrac{1}{a} R\log R]$$
where $a = \|K\| \|Y^{-1}\|$. This shows that 
$$\sup_{R\in \pn}\{ \tr[RK] - \tr[Y f(R)]\} = \sup\{ \tr[RK] - \tr[Y f(R)]\ 
:\ R\geq 0\ ,\ \|R\| \leq e^{1/a}\ \}\ .$$
since the set on the right is compact and convex, and since the function $R 
\mapsto \tr[RK] - \tr[Y f(R)]$ is strictly concave and upper-semicontinuous 
on this set, there exists a unique maximizer, which we have seen must be in 
the interior, and by the strict concavity, there can be no other interior 
critical point. 
\end{proof}

It is now a simple matter to derive the Euler-Lagrange equation that 
determines the maximizer in Lemma~\ref{BSmax}. The integral representation 
for $f(A) = A\log A$ is
$$
A \log A = \int_0^\infty \left( \frac{A}{\lambda +1} - \one + \frac{\lambda 
}{\lambda  +A}\right){\rm d}\lambda\ ,
$$
and then one readily concludes that the unique maximizer $R_{H,Y}$ to the 
variational problem in (\ref{psibsalt}) is the unique solution in $\pn$ of
$$
\int_0^\infty\left( \frac{Y}{\lambda +1}  - \lambda \frac{1}{\lambda +R} Y 
\frac{1}{\lambda +R}\right){\rm d}\lambda = 
Y^{1/2}(H+\one)Y^{1/2}\ .
$$
When $H$ and $Y$ commute, one readily checks that $R = e^H$ is the unique 
solution in $\pn$.

We now show how some of the logarithmic 
inequalities that follow 
from Theorem~\ref{main1} may be used to get upper and lower bounds on  
$\tr[e^{H + \log Y}]$.

Given two positive matrices $W$ and $V$, one way to show that $\tr[W] \leq \tr[V]$ is to show that 
\begin{equation}\label{fulc}
\tr[W \log W]  \leq \tr[W \log V]\ .
\end{equation}
Then 
\begin{eqnarray}\label{fulc2}
0 \leq D(W||V) &=& \tr[W \log W]  - \tr[W \log V] - \tr[W] + \tr[V]\nonumber\\
&\leq& -\tr[W] + \tr[V]\ .
\end{eqnarray}
Thus,  when (\ref{fulc}) is satisfied, one not only has $\tr[W] \leq\ tr[V]$, but the stronger bound $D(W||V) + \tr[W] \leq \tr[V]$.

\begin{thm}\label{HPC2} Let $H,K\in \H_n$  For  $r>0$, define
\begin{equation}\label{comp1}
W :=  (e^{rH} \#_{s}e^{rK})^{1/r}  \quad{\rm and}\qquad V :=  e^{(1-s)H + s 
K}\ .
\end{equation}
Then for $s\in [0,1]$,
\begin{equation}\label{comp2}
D(V||W) + \tr[W] \leq \tr[V]\ .
\end{equation}
\end{thm}

\begin{proof} By the remarks preceding the theorem, it suffices to show that for this choice of $V$ and $W$,  
$\tr[W \log W]  \leq \tr[W \log V]$. Define $X = e^H$ and $Y = 
e^K$. 
The identity
\begin{equation}\label{rmet62AA}
A = (A\#_{s}B)\#_{-s/(1-s)}  B\ , 
\end{equation}
valid for $A,B\in \pn$. 
is the special case  of Theorem~\ref{reparam}  in which $t_1 = 1$, $t = 
-t_0/(t-t_0)$ and $t_0 = s$.
Taking $A = X^r = e^{rH}$ and $B = Y^r = e^{rK}$, we have
$X^r = W^r \#_{\beta} Y^r$, with  $\beta = -s/(1-s)$.
Therefore, by (\ref{secVYV}),
\begin{eqnarray*}
\tr[W\log X]  &=& \frac{1}{r} \tr[W \log (W^r \#_{\beta} Y^r)] =  \\
&\geq &   \tr[W((1-\beta) \log W +  \beta\log Y)].
\end{eqnarray*}
Since $$ \frac{1}{1-\beta} \log X - \frac{\beta}{1-\beta}\log Y = (1-s)\log X + s \log Y = \log V\ ,$$
this last inequality is equivalent to  $\tr[W \log W] \leq  \tr[W \log V]$. 
\end{proof}

\begin{remark}
Since $D(W||V) > 0$ unless $W = V$, (\ref{comp2}) is stronger than the 
inequality 
$\tr[W] \leq \tr[V]$ which is the complemented Golden-Thompson inequality 
of Hiai and Petz \cite{HP93}. Their proof is also based on
(\ref{secVYV}), together with an identity equivalent to (\ref{rmet62AA}), 
but they 
employ these differently, thereby omitting the remainder term
$D(W||V)$.
\end{remark}

We remark that one may obtain at least one of the cases of (\ref{HPco}) directly from (\ref{psibs}) and (\ref{psibs2}) by making an appropriate choice of $X$ in terms of $H$ and $Y$: 
Define $X^{1/2} := Y\# e^H$. Then 
$$X^{1/2}Y^{-1}X^{1/2} = X^{1/2}\#_{-1}Y = e^H\ ,$$
and, therefore, making this choice of $X$,
$$\Psi_{BS}(H,Y) \geq \tr[  (Y\# e^H)^2H] - \tr[(Y\# e^H)^2H] + \tr[(Y\# 
e^H)^2] - \tr[Y] =  \tr[(Y\# e^H)^2] - \tr[Y]\ .$$
This proves  $\tr[(Y\# e^H)^2]    \leq \tr[e^{H +\log Y}] $
which is equivalent to the $r=1/2$, $t=1/2$ case of  \eqref{HPco}.

\section*{Appendices}

\appendix

\section {The Peierls-Bogoliubov Inequality and the \\ Gibbs~Variational 
Principle}  \label{PBGV}

For $A\in \hn$, let $\sigma(A)$ denote the spectrum of $A$, and let $A = \sum_{\lambda\in \sigma(A)}\lambda P_\lambda$
be the spectral decomposition of $A$. For a function $f$ defined $
\sigma(A)$,  $f(A) = \sum_{\lambda\in \sigma(A)}f(\lambda) P_\lambda$.
Likewise, for $B\in \hn$, let  $B = \sum_{\mu\in \sigma(B)}\mu Q_\mu$ be the 
spectral 
decomposition of $B$. 
Let $f$ be convex and differentiable on an interval containing 
$\sigma(A)\cup\sigma(B)$.  
Then, since
$\sum_{\lambda\in \sigma(A)} P_\lambda = \sum_{\mu\in \sigma(B)} Q_\mu= \one$,
\begin{equation}\label{kleinin}
\tr[f(B)- f(A) - f'(A)(B-A)]  = \sum_{\lambda\in \sigma(A)}\sum_{\mu\in 
\sigma(B)} 
[f(\mu) - f(\lambda)  - f'(\lambda)(\lambda - \mu)]\tr[P_\lambda Q_\mu] \ .
\end{equation}
For each $\mu$ and $\lambda$ both $[f(\mu) - f(\lambda)  - 
f'(\lambda)(\lambda - \mu)]$
and $\tr[P_\lambda Q_\mu]$ are non-negative, and hence the right side of 
(\ref{klein}) is non-negative. This yields {\em Klein's inequality}:
\begin{equation}\label{klein}
\tr[f(B)] \geq \tr[f(A)] + \tr[f'(A)(B-A)]\ .
\end{equation}

Now suppose that  the function $f$ is strictly convex on an interval 
containing
$\sigma(A)\cup\sigma(B)$, Then for $\mu\neq \lambda$, $[f(\mu) - f(\lambda)  
- f'(\lambda)(\lambda - \mu)]>0$. If there is equality in (\ref{klein}), 
then for each $\lambda\in \sigma(A)$ and $\mu\in \sigma(B)$ such that 
$\lambda \neq \mu$, $\tr[P_\lambda Q_\mu] =0$. 
Since $\sum_{\mu\in \sigma(B)}\tr[P_\lambda Q_\mu] = \tr[P_\lambda] > 0$, 
$\lambda \in \sigma(B)$
and 
$P_\lambda \leq  Q_\lambda$. The same reasoning shows that for each 
$\mu\in \sigma(B)$, $\mu\in \sigma(A)$ and $Q_\mu \leq P_\lambda$.  Thus, 
there is equality in Klein's inequality if and only if $A = B$.

Taking $f(t) = e^t$, (\ref{klein}) becomes
$\tr[e^{B}] \geq \tr[e^A] + \tr[e^A(B-A)]$.
For $c\in \R$ and $H,K\in \hn$, choose $A = c +H$ and $B = H+K$ to  obtain
$$\tr[e^{H+K}] \geq e^{c}\tr[e^H] + e^{c}\tr[e^H( K- c)]\ .$$
Choosing $c = \tr[e^H K]/\tr[e^H]$, we obtain $\tr[e^{H+K}] \geq e^{\tr[e^H K]/\tr[e^H]}\tr[e^H]$ which can be written as
\begin{equation}\label{PB}
\frac{\tr[e^H K]}{\tr[e^H]} \leq \log(\tr[e^{H+K}]) - \log(\tr[e^H])\ ,
\end{equation}
 the Peierls-Bogoliubov inequality, valid for all $H,K\in \hn$.

The original application of Klein's inequality was to the entropy. It may be used to prove the non-negativity of the relative entropy. Let $A,B\in \pn$, and apply Klein's inequality with $f(x) = x \log x$ to obtain
$$\tr[B \log B] \geq \tr[A\log A] + \tr[(1 + \log A)(B - A)] = \tr[B] - \tr[A] + \tr[B \log A]\ .$$
Rearranging terms yields
$\tr[B(\log B - \log A)] + \tr[A] - \tr[B] \geq 0$;
that is, $D(B||A) \geq 0$.

The Peierls-Bogoliubov Inequality has as a direct consequence  the quantum Gibbs Variational Principle. 
Suppose that $H\in \hn$ and $\tr[e^H] =1$.  Define $X := e^H$ so that $X$ is a density matrix. Then (\ref{PB}) specializes to 
\begin{equation}\label{PB2}
\tr[X K]\leq \log(\tr[e^{\log X+K}]) \ ,
\end{equation}
which is valid for all density matrices $X$ and all $K\in \hn$.  Replacing $K$ in (\ref{PB2}) with 
$K - \log X$ yields
\begin{equation}\label{PB3}
\tr[X K] \leq \log(\tr[e^{K}])  + \tr[X\log X]\ .
\end{equation}
For fixed $X$, there is equality in (\ref{PB3}) for $K = \log X$, and for fixed $K$, there is equality in 
(\ref{PB3}) for $X := e^K/\tr[e^K]$.

It follows that  for all density matrices $X$,
\begin{equation}\label{gibbs1}
\tr[X\log X] =  \sup\{ \tr[XK] - \log(\tr[e^{K}])\ :\ K\in \hn\ \}
\end{equation}
and that for all $K\in \hn$,
\begin{equation}\label{gibbs2}
\log(\tr[e^{K}]) =  \sup\{ \tr[XK] - \tr[X\log X]\ :\ X\in \pn\ , \tr[X] =1\ \}\ .
\end{equation}
This is the Gibbs variational principle for the entropy $S(X) = - \tr[X \log X]$.

Now let $Y\in \pn$ and replace $K$ with 
$K+ \log Y$ in (\ref{PB3}) to conclude that for all density matrices $X$, all $Y\in \pn$ and all $K\in \hn$,
\begin{eqnarray}\label{PB4}
\tr[X K] &\leq& \log(\tr[e^{K+\log Y}])  + \tr[X(\log X - \log Y)] \nonumber\\
&=& (\log(\tr[e^{K+\log Y}])   + 1 -\tr[Y] )+ D(X||Y)\ .
\end{eqnarray}
For fixed $X$, there is equality in (\ref{PB4}) for $K = \log X- \log Y$, and for fixed $K$, there is equality in 
(\ref{PB3}) for $X := e^{K+\log Y}/\tr[e^{K+\log Y}]$. Recalling that for $\tr[X] =1$,  
$\tr[X(\log X - \log Y)] = D(X||Y)+ 1 - \tr[Y]$, 
we have that for all density matrices $X$, and all $Y\in \pn$,
\begin{equation}\label{gibbs3}
D(X||Y) =  \sup\{ \tr[XK] - (\log(\tr[e^{K+ \log Y}]) +\tr[Y] -1 )\ :\ K\in \hn\ \}
\end{equation}
and that for all $K\in \hn$ and all $Y\in \pn$,
\begin{equation}\label{gibbs4}
\log(\tr[e^{K+\log Y}]) + 1 -\tr[Y]  =  \sup\{ \tr[XK] - D(X||Y)\ :\ X\in \pn\ , \tr[X] =1\ \}\ .
\end{equation}

\section{Majorization inequalities}

Let ${\bf x} = (x_1,\dots,x_n)$ and  ${\bf y} = (y_1,\dots,y_n)$  be two vectors in $\R^n$ such that
$x_{j+1} \leq x_j$ and  $y_{j+1} \leq y_j$ for each $j=1,\dots,n-1$.  Then $y$ is said to {\em  majorize} $x$
in case
\begin{equation}\label{maj1}
\sum_{j=1}^k x_j \leq \sum_{j=1}^k y_j \quad{\rm for}\quad k=1,\dots, n-1\quad{\rm and}\quad 
\sum_{j=1}^n x_j = \sum_{j=1}^n y_j \ .
\end{equation}
and in this case we write ${\bf x} \prec {\bf y}$.

A matrix $P\in \mn$ is {\em doubly stochastic} in case 
$P$ has non-negative entries and the entries in each row and column sum to one.
By a theorem of Hardy, Littlewood and P\'olya, ${\bf x} \prec {\bf y}$ if and only if there is a doubly stochastic matrix $P$ such that ${\bf x} = P{\bf y}$.  Therefore, if $\phi$ is convex on $\R$ and ${\bf x}\prec {\bf y}$,
let $P$ be a doubly stochastic matrix such that ${\bf x} = P{\bf y}$. By Jensen's inequality
$$\sum_{j=1}^n \phi(x_j) = \sum_{j=1}^n \phi\left(\sum_{k=1}^n P_{j,k} y_k \right)  \leq 
\sum_{j,k=1}^n P_{j,k}\phi\left( y_k \right)  = \sum_{k=1}^n \phi(y_k)\ . $$
That is, for every convex function $\phi$,
\begin{equation}\label{maj3}
{\bf x} \prec {\bf y} \quad{\Rightarrow} \quad \sum_{j=1}^n \phi(x_j) \leq  \sum_{j=1}^n \phi(y_j)\ .
\end{equation}

Let $X,Y\in \hn$, and let  ${\boldsymbol\lambda}^X$ and ${\boldsymbol\lambda}^Y$ be the eigenvalue sequences of $X$ and $Y$ respectively with the eigenvalues repeated according to their geometric multiplicity and arranged in decreasing order considered as vectors in $\R^n$. Then $Y$ is said to {\em majorize} $X$ in case 
${\boldsymbol\lambda}^X \prec{\boldsymbol\lambda}^Y$, and in this case we write $X \prec  Y$.
It follows immediately from (\ref{maj3}) that if $\phi$ is an  increasing convex function,
\begin{equation}\label{maj4}
X \prec Y \quad \Rightarrow \quad \tr[\phi(X)] \leq \tr[\phi(Y)] \quad{\rm and}\quad \tr[X] = \tr[Y]\ .
\end{equation}
The following extends a theorem of Bapat and Sunder \cite{BS85}:

\begin{thm}\label{bohpatsunder}   Let $\Phi:\mn \to \mn$ be a linear transformation such that $\Phi(A) \geq 0$ for all $A \geq 0$, $\Phi(\one) = \one$ and $\tr[\Phi(A)] = \tr[A]$ for all $A \in \mn$.  Then for all $A\in \hn$,
\begin{equation}\label{bosu2}
\Phi(X) \prec X\ .
\end{equation}
\end{thm}

\begin{proof} Note that $\Phi(X)\in \hn$.  Let $\Phi(X) = \sum_{j=1}^n \lambda_j |v_j\rangle\langle  v_j|$ be the spectral resolution of $\Phi(X)$ with $\lambda_j \geq \lambda_{j+1}$ for $j =1,\dots,n-1$,
Fix $k\in \{1,\dots,n-1\}$.and let $P_k =   \sum_{j=1}^k  |v_j\rangle\langle v_j|$. Then with $\Phi^*$ denoting the adjoint of $\Phi$ with respect to the Hilbert-Schmidt inner product, 
\begin{eqnarray*}
\sum_{j=1}^k \lambda_j &=& \tr[P_k\Phi(X)]\\
&=& \tr[\Phi^*(P_k)X] \leq \sup\{ \tr[QX], \ 0 \leq Q \leq 1,\tr[Q] = k\ \} = 
\sum_{j=1}^k \mu_j\ 
\end{eqnarray*}
where $\{\mu_1,\dots,\mu_k\}$ is the eigenvalue sequence of $X$ arranged in decreasing order. 
\end{proof}

Bapat and Sunder prove this for $\Phi$ of the form $\Phi(A) = \sum_{\j=1}^m V_j^*AV_j$ where
Let $V_1,\dots, V_m \in \mn$ satisfy
\begin{equation}\label{bosu1}
\sum_{j=1}^m V_j V_j^* = \one  = \sum_{j=1}^m V_j^* V_j\ .
\end{equation}
Choi \cite{Choi72,Choi75} has shown that, for all $n\geq 2$, the transformation
$$\Phi(A) = \frac{1}{n^2 - n -1}((n-1)\tr[A]\one - A)$$
cannot be written in the form (\ref{bosu1}), yet it  satisfies the conditions of Theorem~\ref{bohpatsunder}.

\begin{lm}\label{bosulem} Let $A\in \pn$ and let $\Phi$ be defined by 
\begin{equation}\label{bosu3}
\Phi(X) = \int_0^\infty \frac{ A^{1/2}}{\lambda + A} X \frac{ A^{1/2}}{\lambda + A}  {\rm d}\lambda
\end{equation}. Then for all $X\in \hn$, (\ref{bosu2}) is satisfied,
and for all $p\geq 1$, 
\begin{equation}\label{bosu5}
\tr[|\Phi(X)|^p] \leq \tr[|X|^p] \ .
\end{equation}
\end{lm}

\begin{proof} $\Phi$ evidently satisfies the conditions of Theorem~\ref{bohpatsunder}, and then (\ref{bosu2}) implies (\ref{bosu5}) as discussed above. 
\end{proof}

\section{Geodesics and Geometric Means}\label{geodesic}

There is 
a natural Riemannian metric on $\pn$
such that the corresponding distance $\delta(X,Y)$ is invariant under 
conjugation:
$$\delta(A^*XA,A^*YA) = \delta(X,Y)$$
for all $X,Y\in \pn$ and all invertible $n\times n$ matrices $A$. It turns out that for  
$A,B\in \pn$, $t\mapsto A\#_t B$, $t\in [0,1]$, is a constant speed geodesic for this metric that connects $A$ and $B$. 
This geometric point of view, originating in the work of statisticians, and was developed in the form presented here by Bhatia and Holbrook \cite{BH06}.

\begin{defi} Let $t\mapsto X(t)$, $t\in [a,b]$,  be a smooth path in $\pn$. The arc-length along this path
in the {\em conjugation invariant metric} is
$$\int_a^b \|X(t)^{-1/2} X'(t) X(t)^{-1/2}\|_2{\rm d} t\ ,$$
where $\|\cdot\|_2$ denotes the Hilbert-Schmidt norm and the prime denotes the derivative.
The corresponding distance between $X,Y\in \pn$ is defined by 
$$\delta(X,Y) = \inf\left \{  \int_0^1  \|X(t)^{-1/2} X'(t) X(t)^{-1/2}\|_2{\rm d} t\ : X(t)\in \pn \ {\rm  for}\  t\in (0,1), X(0) =X, X(1) = Y\ \right\}\ .$$
\end{defi}

To see the conjugation invariance, let the smooth path $X(t)$ be given, let an invertible matrix $A$ be given, and define $Z(t) := A^*X(t)A$. Then by cyclicity of the trace,
\begin{eqnarray*}
\|Z(t)^{-1/2} Z'(t) Z(t)^{-1/2}\|_2^2 &=& \tr[ Z(t)^{-1} Z'(t) Z(t)^{-1} Z'(t)] \\
&=&
\tr[ A^{-1}X(t)^{-1} X'(t) X(t)^{-1} X'(t)A] = \|X(t)^{-1/2} X'(t) X(t)^{-1/2}\|_2^2\ .
\end{eqnarray*}

Given any smooth path $t \mapsto X(t)$, define $H(t) := \log(X(t))$ so that $X(t)= e^{H(t)}$, and then
\begin{equation}\label{rmet1}
X'(t) = \int_0^1 X(t)^{1-s} H'(t) X(t)^s {\rm d} s\ ,
\end{equation}
or equivalently,
\begin{eqnarray}\label{rmet1}
H '(t) &=& \int_0^\infty \frac{1}{\lambda + X(t)} X'(t) \frac{1}{\lambda + X(t)} {\rm d} \lambda\nonumber\\
&=&  \int_0^\infty \frac{X(t)^{1/2}}{\lambda + X(t)} (X(t)^{-1/2}X'(t)X(t)^{-1/2}) \frac{X(t)^{1/2}}{\lambda + X(t)} {\rm d} \lambda\ .
\end{eqnarray}

Lemma~\ref{bosulem} yields 
$H'(t) \prec X(t)^{-1/2}X'(t)X(t)^{-1/2}$
and  its consequence
\begin{equation}\label{rmet3}
\|H'(t)\|_2  \leq  \|X(t)^{-1/2}X'(t)X(t)^{-1/2}\|_2 \ .
\end{equation}

Now let $X(t)$ be a smooth path in $\pn$ with $X(0) = X$ and $X(1) = Y$. Then, with $H(t) = \log X(t)$
\begin{eqnarray}\label{rmet4}
\|\log Y - \log X\|_2  &=& \left\Vert \int_0^1 H'(t){\rm d} t\right\Vert_2\nonumber\\
 &\leq& 
 \int_0^1\| H'(t)\|_2{\rm d}t \leq  \int_0^1\|X(t)^{-1/2}X'(t)X(t)^{-1/2}\|_2{\rm d}t = \delta(X,Y)\ .
 \end{eqnarray}
 
 If $X$ and $Y$ commute, this lower bound is exact: Given $X,Y\in \pn$ that commute, define
 $H(t) = (1-t)\log X + t \log Y$, and $X(t) = e^{H(t)}$.  Then $H'(t) = \log Y - \log X$, independent of $t$.
 Hence all of the inequalities in (\ref{rmet4}) are equalities. Moreover, if  there is equality in (\ref{rmet4}),
 the necessarily  
 $$H'(s) = \int_0^1 H'(t){\rm d} t = \log Y - \log X$$
 for all $s\in [0,1]$.  
 This proves:

  \begin{lm}\label{comcase}
  When $X,Y \in \pn$ commute, there is exactly one constant speed geodesic running from $X$ to $Y$ in unit time, namely, $X(t) = e^{(1-t)\log X + t\log Y}$, and 
 $$\delta(X,Y) =  \|\log Y - \log X\|_2\ .$$
 \end{lm}
 
  Since conjugation is an isometry in this metric, it is now a simple matter to find the explicit formula for the geodesic connecting $X$ and $Y$ in $\pn$.  Apart from the statement on uniqueness, the following theorem is due to Bhatia and Holbrook \cite{BH06}.

\begin{thm}\label{gencase}  For all $X,Y\in \pn$, 
there is exactly one constant speed geodesic running from $X$ to $Y$ in unit time, namely, 
\begin{equation}\label{rmet5}
 X(t) = X\#_t Y := X^{1/2}(X^{-1/2}Y X^{-1/2})^tX^{1/2}\ ,
 \end{equation} 
 and
 $$\delta(X,Y) =  \| \log (X^{-1/2}Y X^{-1/2})\|_2\ .$$
 \end{thm}
 
 \begin{proof}  By Lemma~\ref{comcase}, the unique 
 constant speed geodesic running from $\one$ to $X^{-1/2}Y X^{-1/2}$ in unit time
 is $W(t) = (X^{-1/2}Y X^{-1/2})^t$; it has the constant speed $\|\log (X^{-1/2}Y X^{-1/2})\|_2$, 
and 
$$\delta(\one, X^{-1/2}Y X^{-1/2}) = \|\log (X^{-1/2}Y X^{-1/2})\|_2 = \delta(\one, X^{-1/2}Y X^{-1/2})\ .$$
 By the conjugation invariance of
the metric, $\delta(X,Y) = \delta(\one, X^{-1/2}Y X^{-1/2})$ and $X(t)$ as defined in (\ref{rmet5}) 
has the constant speed $\delta(X,Y)$ and runs from $X$ to $Y$ in unit time. Thus it is a constant 
speed geodesic running from $X$ to $Y$ in unit time.   

If there were another such geodesic, say $\widetilde{X}(t)$,
then $X^{-1/2}\widetilde{X}(t) X^{-1/2}$ would be a 
constant speed geodesic running from $\one$ to $X^{-1/2}Y X^{-1/2}$ in unit time, 
and different form $W(t)$, but this would contradict the uniqueness in Lemma~\ref{comcase}. 
 \end{proof}

 In particular, the midpoint of the unique constant speed geodesic running from 
 $X$ to $Y$ in unit time  is the geometric mean of $X$ and $Y$ as originally defined by Pusz and Woronowicz \cite{PW75}:
 $$X\# Y = X^{1/2}(X^{-1/2}Y X^{-1/2})^{1/2}X^{1/2}\ .$$
  
In fact, the Riemannian manifold $(\pn, \delta)$ is geodesically complete: The smooth path
$$t\mapsto X^{1/2}(X^{-1/2}Y X^{-1/2})^tX^{1/2} := X\#_t Y$$
is well defined for all $t\in \R$. By the conjugation invariance and Lemma~\ref{comcase}, for all $s,t\in \R$,
$$\delta(X\#_s Y, X\#_t Y) = \delta((X^{-1/2}Y X^{-1/2})^s, (X^{-1/2}Y X^{-1/2})^t)  = |t-s|\| \log (X^{-1/2}Y X^{-1/2})\|_2\ .$$
Since the speed along the curve $T\mapsto  X\#_t Y$ has the constant value $\|\log (X^{-1/2}Y X^{-1/2})\|_2$,
this, together with the uniqueness in Theorem~\ref{gencase}, shows that    for all $t_0 < t_1$ in $\R$,
the restriction of $t\mapsto X\#_t Y$ to $[t_0,t_1]$ is the unique constant speed geodesic running from 
$X\#_{t_0} Y$ to $X\#_{t_1} Y$ in time $t_1- t_0$.  

This has a number of consequences. 

\begin{thm}\label{reparam}  Let $X,Y\in \pn$, and $t_0 , t_1\in \R$. Then for all $t\in \R$
\begin{equation}\label{rmet62}
X\#_{(1-t)t_0 + t t_1} Y = (X\#_{t_0}Y)\#_t   (X\#_{t_1}Y)
\end{equation}
\end{thm}

\begin{proof} By what we have noted above, $t\mapsto X\#_{(1-t)t_0 + t t_1} Y$ is a constant speed geodesic running from $X\#_{t_0}Y$ to $X\#_{t_1}Y$ in unit time, as is $t\mapsto 
(X\#_{t_0}Y)\#_t   (X\#_{t_1}Y)$. The identity (\ref{rmet62}) now follows from the uniqueness in 
Theorem~\ref{gencase}.  
\end{proof}

Taking $t_0 = 0$ and $t_1 =s$, we have the special case
\begin{equation}\label{rmet62A}
X\#_{ts} Y = X\#_t   (X\#_sY)\ .
\end{equation}

Taking $t_0 = 1$ and $t_1 =0$, we have the special case
\begin{equation}\label{rmet62B}
X\#_{1-t} Y = Y\#_t  X\ .
\end{equation}
The identity (\ref{rmet62B}) is well-known, and may be derived directly from the formula in (\ref{rmet5}).

We are particularly concerned with $t\mapsto X\#_t Y$ for $t\in [-1,2]$. Indeed,  from the formula in 
(\ref{rmet5}),
\begin{equation}\label{geom2}
X\#_{-1} Y = X \frac{1}{Y} X \qquad{\rm and}\qquad X\#_2 Y = Y\frac{1}{X} Y\ .
\end{equation}

Let $t\in (0,1)$. By combining the formula
$$X\#_tY = X^{1/2}(X^{-1/2}Y X^{-1/2})^t X^{1/2} =  X^{1/2}(X^{1/2}Y^{-1} X^{1/2})^{-t} X^{1/2}$$
with the integral representation
$$A^{-t} = \frac{\sin(\pi t)}{\pi}\int_0^\infty \lambda^{-t} \frac{1}{\lambda + A} {\rm d}\lambda = 
\frac{\sin(\pi t)}{\pi}\int_0^\infty \lambda^{t} \frac{1}{1 + \lambda A} {\rm d}\lambda\ , $$
we obtain, for $t\in (0,1)$, 
\begin{eqnarray}\label{akrep}
X\#_tY  &=&   \frac{\sin(\pi t)}{\pi}\int_0^\infty \lambda^{t} X^{1/2}\frac{1}{1 + \lambda X^{1/2}Y^{-1} X^{1/2}}X^{1/2} {\rm d}\lambda\nonumber\\
&=&   \frac{\sin(\pi t)}{\pi}\int_0^\infty \lambda^{t} \frac{1}{X^{-1} + \lambda Y^{-1} } {\rm d}\lambda\ .
\end{eqnarray}
The merit of this formula lies in the following lemma:
\begin{lm}[Ando]\label{aklem}
The function $(A,B) \mapsto (A^{-1} + B^{-1})^{-1}$ is jointly concave on $\pn$. 
\end{lm}

\begin{proof} Note that
$A^{-1} + B^{-1} = A^{-1}(A+ B)B^{-1}$, so that
$$(A^{-1} + B^{-1})^{-1}  = B(A+B)^{-1}A = ((A+B) -A) (A+B)^{-1}A = A -  A(A+B)^{-1}A\ ,$$
and the claim now follows form the convexity of $(A,B) \mapsto A(A+B)^{-1}A$ \cite{K59}. 
\end{proof} 

The {\em harmonic mean} of positive operators $A$ and $B$, $A:B$, is defined by 
\begin{equation}\label{harme}
A:B := 2(A^{-1} + B^{-1})^{-1}\ ,
\end{equation}
and hence Lemma~\ref{aklem} says that $(A,B) \mapsto A:B$ is jointly concave. Moreover, (\ref{akrep}) can be written in terms of the harmonic mean as
\begin{equation}\label{akrep2}
X\#_tY =  \frac{\sin(\pi t)}{2\pi} \int_0^\infty  X:(\lambda Y) \lambda^{t} {\rm d}\lambda\ ,
\end{equation}
which expresses weighted geometric means as average over harmonic means.
By the operator  monotonicity of the map $A\mapsto A^{-1}$, the map $X,Y \mapsto X:Y$ is monotone in each variable, and then by (\ref{akrep2}) this is also true of $X,Y \mapsto X\#_t Y$. This proves the following result of Ando and Kubo \cite{KA80}:

\begin{thm}[Ando and Kubo]\label{akthm} For all $t\in [0,1]$, 
$(X,Y)\mapsto X\#_tY$ is jointly concave, and monotone increasing in $X$ and $Y$.
\end{thm}

The method of Ando and Kubo can be used to prove joint operator concavity theorems for functions on $\pn\times \pn$  that are not connections.  The next theorem, due to Fujii and Kamei \cite{FK88}, provides an important example. 

\begin{thm}\label{jointconvex} The map
$(X,Y) \mapsto  -X^{1/2} \log(X^{1/2} Y^{-1} X^{1/2}) X^{1/2}$ is jointly concave. 
\end{thm}

\begin{proof} The representation
$$\log A = \int_0^\infty \left(
\frac{1}{\lambda + 1} - \frac{1}{\lambda + A}\right){\rm d}\lambda$$
yields
$$- X^{1/2} \log(X^{1/2} Y^{-1} X^{1/2}) X^{1/2} = 
\int_0^\infty \left(
 \frac{1}{ X^{-1} + (\lambda Y)^{-1}}   - \frac{1}{\lambda + 1}X \right){\rm d}\lambda$$
from which the claim follows. 
\end{proof}

\begin{thm}\label{jointconvexout} For all $t\in [-1,0]\cup [1,2]$, the map $(X,Y) \mapsto X\#_tY$ is jointly convex. 
\end{thm}

\begin{proof} First suppose that $t\in [0,1]$. The case $t=0$ is trivial, and since $X\#_{-1}Y = XY^{-1}X$ which is convex, we may suppose that $t\in (-1,0)$.  Let $s= -t$ so that $s\in (0,1)$. 
We use the integral representation
$$A^ s = \frac{\sin\pi s}{\pi} \int_0^\infty \lambda^s\left( \frac{1}{\lambda} - \frac{1}{\lambda + A}\right){\rm d}\lambda$$
valid for $A\in \pn$ and $s\in (0,1)$
to obtain
$$X\#_sY   = 
\frac{\sin\pi s}{\pi} \int_0^\infty \lambda^s \left( X  - \frac{1}{ X^{-1}  + (\lambda Y)^{-1}}\right)\frac{{\rm d}\lambda}{\lambda}\ ,$$
which by Lemma~\ref{aklem} is jointly convex. Finally, the identity $Y\#_{1-t}X = X\#_tY$ shows that the joint convexity for $t\in [1.2]$ follows from the joint convexity for $t\in [-1,0]$.
\end{proof}

\end{document}